\documentclass[11pt]{article}%
\usepackage{amsfonts}
\usepackage{amsmath}
\usepackage{fullpage}
\usepackage{amssymb}
\usepackage{graphicx}
\usepackage{hyperref}%
\providecommand{\U}[1]{\protect\rule{.1in}{.1in}}
\pdfoutput=1
\newtheorem{theorem}{Theorem}

\newtheorem{corollary}[theorem]{Corollary}

\newtheorem{lemma}[theorem]{Lemma}

\newtheorem{remark}[theorem]{Remark}

\newenvironment{proof}[1][Proof]{\noindent\textbf{#1.} }{\ \rule{0.5em}{0.5em}}
\numberwithin{equation}{section}
\newcommand*{\id}
{\mathrm{id}}
\newcommand*{\tr}
{\mathrm{tr}}

\begin{document}

\title{Monotonicity of quantum relative entropy and recoverability}
\author{Mario Berta\thanks{Institute for Quantum Information and Matter, California
Institute of Technology, Pasadena, California 91125, USA}
\and Marius Lemm\thanks{Mathematics Department, California Institute of Technology,
Pasadena, California 91125, USA}
\and Mark M. Wilde\thanks{Hearne Institute for Theoretical Physics, Department of
Physics and Astronomy, Center for Computation and Technology, Louisiana State
University, Baton Rouge, Louisiana 70808, USA}}
\maketitle

\begin{abstract}
The relative entropy is a principal measure of distinguishability in quantum
information theory, with its most important property being that it is
non-increasing with respect to noisy quantum operations. Here, we establish a remainder
term for this inequality that quantifies how well one can recover from a loss
of information by employing a rotated Petz recovery map. The main approach for
proving this refinement is to combine the methods of [Fawzi and Renner,
arXiv:1410.0664] with the notion of a relative typical subspace from
[Bjelakovic and Siegmund-Schultze, arXiv:quant-ph/0307170]. Our paper constitutes partial progress towards a remainder term which
features just the Petz recovery map (not a rotated Petz map), a conjecture
which would have many consequences in quantum information
theory.

A well known result states that the monotonicity
of relative entropy with respect to quantum operations is equivalent to each of the
following inequalities: strong subadditivity of entropy, concavity of
conditional entropy, joint convexity of relative entropy, and monotonicity of
relative entropy with respect to partial trace. We show that this equivalence holds true
for refinements of all these inequalities in terms of the Petz recovery map.
So either all of these refinements are true or all are false.

\end{abstract}

\section{Introduction}

The Umegaki relative entropy $D\left(  \rho\Vert\sigma\right)$ between a
density operator\footnote{Recall that a density operator is a positive
semi-definite operator with trace equal to one. Throughout this paper,
sometimes our statements apply only to positive definite density operators,
and we make it clear when this is so.} $\rho$ and a positive semi-definite
operator $\sigma$ is defined as
 Tr$\left\{  \rho\left[  \log\rho-\log\sigma\right]  \right\}  $ whenever $\operatorname{supp}(\rho)
 \subseteq 
 \operatorname{supp}(\sigma)$ and as $+\infty$ otherwise. It is a fundamental information measure in quantum information
theory \cite{U62}, from which many other information measures, such as
entropy, conditional entropy, and mutual information, can be derived (see,
e.g., \cite{BSW14}). When $\sigma$ is a density operator, the relative entropy
is a measure of statistical distinguishability and receives an operational
interpretation in the context of asymmetric quantum hypothesis testing (known
as the quantum Stein's lemma) \cite{HP91,NO00}. Being a good measure of
distinguishability, the relative entropy does not increase with respect to quantum
processing, as is captured in the following inequality, known as monotonicity
of relative entropy \cite{Lindblad1975,U77}:%
\begin{equation}
D\left(  \rho\Vert\sigma\right)  \geq D\left(  \mathcal{N}\left(  \rho\right)
\Vert\mathcal{N}\left(  \sigma\right)  \right)  ,\label{eq:mono-rel-ent}%
\end{equation}
where $\mathcal{N}$ is a linear completely positive trace preserving (CPTP)\ map
(also referred to as a quantum channel). The inequality is known to be
saturated if and only if the following Petz recovery map perfectly recovers
$\rho$ from $\mathcal{N}\left(  \rho\right)  $ \cite{Petz1986,Petz1988} (see
also \cite{HJPW04}):%
\begin{equation}
\mathcal{R}_{\sigma,\mathcal{N}}^{P}\left(  \cdot\right)  \equiv\sigma
^{1/2}\mathcal{N}^{\dag}\left[  \left(  \mathcal{N}\left(  \sigma\right)
\right)  ^{-1/2}\left(  \cdot\right)  \left(  \mathcal{N}\left(
\sigma\right)  \right)  ^{-1/2}\right]  \sigma^{1/2},\label{eq:Petz-map}%
\end{equation}
with $\mathcal{N}^{\dag}$ the adjoint of $\mathcal{N}$. (Observe that the Petz
recovery map always perfectly recovers $\sigma$ from $\mathcal{N}\left(
\sigma\right)  $.) There are several related inequalities, which are known to
be equivalent\footnote{The notion that two statements which are known to be true are
`equivalent' of course does not make strict sense logically. So when we
say that `$A$ is equivalent to $B$' for two statements $A$ and $B$ which are
already known to be true (for us $A$ and $B$ will always be some kind of
entropy inequalities), we in fact mean the softer (but standard) notion
that, if one assumes $A$, then there exists a relatively direct proof for $B$ and vice versa.} to (\ref{eq:mono-rel-ent}) when $\sigma$ is a density operator (see,
e.g., \cite{R02}). One equivalent inequality is the monotonicity of relative
entropy with respect to partial trace:%
\begin{equation}
D\left(  \rho_{AB}\Vert\sigma_{AB}\right)  \geq D\left(  \rho_{B}\Vert
\sigma_{B}\right)  ,\label{eq:mono-partial}%
\end{equation}
where $\rho_{AB}$ and $\sigma_{AB}$ are density operators acting on a
tensor-product Hilbert space $\mathcal{H}_{A}\otimes\mathcal{H}_{B}$. The
operators $\rho_{B}$ and $\sigma_{B}$ result from the partial trace:\ $\rho
_{B}\equiv\operatorname{Tr}_{A}\left\{  \rho_{AB}\right\}  $ and $\sigma
_{B}\equiv\operatorname{Tr}_{A}\left\{  \sigma_{AB}\right\}  $. Another
equivalent inequality is the joint convexity of relative entropy:%
\begin{equation}
\sum_{x}p_{X}\left(  x\right)  D\left(  \rho^{x}\Vert\sigma^{x}\right)  \geq
D\left(  \overline{\rho}\Vert\overline{\sigma}\right)
,\label{eq:joint-convex}%
\end{equation}
where $p_{X}$ is a probability distribution, $\left\{  \rho^{x}\right\}  $ and
$\left\{  \sigma^{x}\right\}  $ are sets of density operators, $\overline
{\rho}\equiv\sum_{x}p_{X}\left(  x\right)  \rho^{x}$, and $\overline{\sigma
}\equiv\sum_{x}p_{X}\left(  x\right)  \sigma^{x}$. The interpretation of the
above inequality is that distinguishability cannot increase under the loss of
the classical label$~x$. One other equivalent inequality is the strong
subadditivity of quantum entropy \cite{PhysRevLett.30.434,LR73}:%
\begin{equation}
I\left(  A;B|C\right)  _{\omega}\equiv D\left(  \omega_{ABC}\Vert\omega
_{AC}\otimes I_{B}\right)  -D\left(  \omega_{BC}\Vert\omega_{C}\otimes
I_{B}\right)  \geq0,\label{eq:SSA}%
\end{equation}
which can be seen as a special case of (\ref{eq:mono-rel-ent}) with
$\rho=\omega_{ABC}$, $\sigma=\omega_{AC}\otimes I_{B}$, and $\mathcal{N}%
=\operatorname{Tr}_{A}$, where $\omega_{ABC}$ is a tripartite density operator
acting on the tensor-product Hilbert space $\mathcal{H}_{A}\otimes
\mathcal{H}_{B}\otimes\mathcal{H}_{C}$. A final equivalent inequality that we
mention is the concavity of conditional entropy \cite{LR73}:%
\begin{equation}
H\left(  A|B\right)  _{\overline{\rho}}\geq\sum_{x}p_{X}\left(  x\right)
H\left(  A|B\right)  _{\rho^{x}},\label{eq:concavity}%
\end{equation}
where $p_{X}$ is a probability distribution, $\{\rho_{AB}^{x}\}$ is a set of
density operators, $\overline{\rho}_{AB}\equiv\sum_{x}p_{X}\left(  x\right)
\rho_{AB}^{x}$, and the conditional entropy $H\left(  A|B\right)  _{\sigma
}\equiv-D\left(  \sigma_{AB}\Vert I_{A}\otimes\sigma_{B}\right)  $.

The above inequalities have been critical to the development of quantum
information theory. In fact, since so much of quantum information theory
relies on these inequalities and given that they are equivalent and apply
universally for any states and channels, they are often considered to
constitute a fundamental law of quantum information theory. In light of this,
we might wonder if there could be refinements of the above inequalities in the
form of \textquotedblleft remainder terms.\textquotedblright\ While a number
of works pursued this direction
\cite{BCY11,Winterconj,K13conj,LW14a,CL14,Z14,Z14b,BSW14,SBW14,LW14,SW14}, a
breakthrough paper established the following remainder term for strong
subadditivity \cite{FR14}:%
\begin{equation}
I\left(  A;B|C\right)  _{\omega} \geq-\log F\left(  \omega_{ABC},\left(
\mathcal{V}_{AC}\circ\mathcal{R}_{C\rightarrow AC}^{P}\circ\mathcal{U}%
_{C}\right)  \left(  \omega_{BC}\right)  \right)  , \label{eq:FR}%
\end{equation}
where $F\left(  \tau,\varsigma\right)  \equiv\left\Vert \sqrt{\tau}%
\sqrt{\varsigma}\right\Vert _{1}^{2}$ is the quantum fidelity between positive
semi-definite operators $\tau$ and $\varsigma$ \cite{U73}, $\mathcal{U}_{C}$
and $\mathcal{V}_{AC}$ are unitary channels defined in terms of some unitary
operators $U_{C}$ and $V_{AC}$ as%
\begin{align}
\mathcal{U}_{C}\left(  \cdot\right)   &  \equiv U_{C}\left(  \cdot\right)
U_{C}^{\dag},\\
\mathcal{V}_{AC}\left(  \cdot\right)   &  \equiv V_{AC}\left(  \cdot\right)
V_{AC}^{\dag},
\end{align}
and $\mathcal{R}_{C\rightarrow AC}^{P}$ is the following Petz recovery map:%
\begin{equation}
\mathcal{R}_{C\rightarrow AC}^{P}\left(  \cdot\right)  \equiv\omega_{AC}%
^{1/2}\omega_{C}^{-1/2}\left(  \cdot\right)  \omega_{C}^{-1/2}\omega
_{AC}^{1/2}.
\end{equation}

In the present paper, our first contribution is to combine the methods of
\cite{FR14}\ and the notion of a relative typical subspace from \cite[pages
4-5]{BS12}\ in order to establish the following remainder term for the
inequality in (\ref{eq:mono-rel-ent}):%
\begin{equation}
D\left(  \rho\Vert\sigma\right)  -D\left(  \mathcal{N}\left(  \rho\right)
\Vert\mathcal{N}\left(  \sigma\right)  \right)  \geq-\log F\left(
\rho,\left(  \mathcal{V\circ R}_{\sigma,\mathcal{N}}^{P}\circ\mathcal{U}%
\right)  \left(  \mathcal{N}\left(  \rho\right)  \right)  \right)  ,
\label{eq:refinement-ineq}%
\end{equation}
where $\mathcal{U}$ is a unitary channel acting on the output space of
$\mathcal{N}$, $\mathcal{R}_{\sigma,\mathcal{N}}^{P}$ is the Petz recovery map
defined in (\ref{eq:Petz-map}), and $\mathcal{V}$ is a unitary channel acting
on the input space of $\mathcal{N}$. Thus, the refinement in
(\ref{eq:refinement-ineq})\ quantifies how well one can recover $\rho$ from
$\mathcal{N}\left(  \rho\right)  $ by employing the \textquotedblleft rotated
Petz recovery map\textquotedblright\ $\mathcal{V\circ R}_{\sigma,\mathcal{N}%
}^{P}\circ\mathcal{U}$. This result is stated formally as
Corollary~\ref{cor:mono} and can be understood as a generalization of
(\ref{eq:FR}).
We establish a similar refinement of the inequality in
(\ref{eq:mono-partial}), stated formally as
Theorem~\ref{thm:mono-PT}. Given that the original inequalities without remainder terms
have found wide use in quantum information theory, we expect the refinements with
remainder terms presented here to find use in some applications of the
original inequalities, perhaps in the context of quantum error correction
\cite{Barnum2002,Schumacher2002,Tyson2010,Ng2010,Mandayam2012} or
thermodynamics \cite{RevModPhys.74.197,S12}. Note that the refinement in
(\ref{eq:FR}) has already been helpful in improving our understanding of some
quantum correlation measures \cite{Winterconj,LW14,SW14,W14}.

It would be very useful for applications if the aforementioned refinements of
relative entropy inequalities held for the Petz recovery map (and not merely
for a rotated Petz recovery map), i.e., if they were of the following form:
\begin{align}
D\left(  \rho\Vert\sigma\right)  -D\left(  \mathcal{N}\left(  \rho\right)
\Vert\mathcal{N}\left(  \sigma\right)  \right)   &  \geq-\log F\left(
\rho,\mathcal{R}_{\sigma,\mathcal{N}}^{P}\left(  \mathcal{N}\left(
\rho\right)  \right)  \right)  ,\label{eq:conj-1}\\
D\left(  \rho_{AB}\Vert\sigma_{AB}\right)  -D\left(  \rho_{B}\Vert\sigma
_{B}\right)   &  \geq-\log F\left(  \rho_{AB},\sigma_{AB}^{1/2}\sigma
_{B}^{-1/2}\rho_{B}\sigma_{B}^{-1/2}\sigma_{AB}^{1/2}\right)  ,\\
\sum_{x}p_{X}( x) D( \rho^{x}\Vert\sigma^{x}) -D( \overline{\rho}%
\Vert\overline{\sigma})  &  \geq-2\log\sum_{x}p_{X}( x) \sqrt{F( \rho
^{x},\left(  \sigma^{x}\right)  ^{\frac{1}{2}}\left(  \overline{\sigma
}\right)  ^{-\frac{1}{2}}\overline{\rho}\left(  \overline{\sigma}\right)
^{-\frac{1}{2}}\left(  \sigma^{x}\right)  ^{\frac{1}{2}}) },\\
I(A;B|C)_{\omega}  &  \geq-\log F\left(  \omega_{ABC},\omega_{AC}^{1/2}%
\omega_{C}^{-1/2}\omega_{BC}\omega_{C}^{-1/2}\omega_{AC}^{1/2}\right)  ,\\
H\left(  A|B\right)  _{\overline{\rho}}-\sum_{x}p_{X}\left(  x\right)
H\left(  A|B\right)  _{\rho^{x}}  &  \geq-2\log\sum_{x}p_{X}(x)\sqrt{F\left(
\rho_{AB}^{x},\overline{\rho}_{AB}^{1/2}\overline{\rho}_{B}^{-1/2}\rho_{B}%
^{x}\overline{\rho}_{B}^{-1/2}\overline{\rho}_{AB}^{1/2}\right)  }.
\label{eq:conj-3}%
\end{align}
In \cite[Definition~25]{SBW14}, a R\'{e}nyi information measure was defined to
generalize relative entropy differences. The inequalities
\eqref{eq:conj-1}-\eqref{eq:conj-3} stated above would follow from the
monotonicity of this R\'{e}nyi information measure with respect to the
R\'{e}nyi parameter (see \cite[Conjecture~26]{SBW14}, \cite[Consequences 27
and 28]{SBW14}). A weaker form of (\ref{eq:conj-1}) in terms of trace distance
on the right-hand side was first conjectured in \cite[Eq.~(4.7)]{Z14b}.

Our second contribution in this paper is to show that slightly weaker forms
of these inequalities, featuring instead the square of the Bures distance
\cite{Bures1969} $D_{B}^{2}\left(  \omega,\tau\right)  \equiv2(1-\sqrt
{F\left(  \omega,\tau\right)  })$ on the right-hand side, are all equivalent
(observe that $-\log(F)\geq2(1-\sqrt{F})$). So either all of these refinements
are true or all are false. It remains an important open question to determine
which is the case. This second contribution is in principle conjectural, but we believe it is nonetheless important, for two reasons: (1) Obviously, it reduces the work of proving (or even disproving) entropy inequalities with Petz remainder terms to single cases, which can be chosen according to convenience. (2) It furthers the evidence that the Petz remainder term is the natural one.

The next section recalls the notion of a relative typical subspace and the
remaining sections give proofs of our claims.

\section{Relative typical subspace}

\label{sec:stein-review}We begin by reviewing the notion of a relative typical
subspace from \cite[pages 4-5]{BS12}. Consider spectral decompositions of a
density operator $\rho$ and a positive semi-definite operator $\sigma$ acting
on a finite-dimensional Hilbert space, such that $\operatorname{supp}%
(\rho)\subseteq\operatorname{supp}(\sigma)$:%
\begin{align}
\rho &  =\sum_{x}p_{X}(x)|\psi_{x}\rangle\langle\psi_{x}|,\\
\sigma &  =\sum_{y}f_{Y}(y)|\phi_{y}\rangle\langle\phi_{y}|.
\end{align}
Let us define the relative typical subspace $T_{\rho|\sigma}^{\delta,n}$ for
$\delta>0$ and integer $n\geq1$ as%
\begin{equation}
T_{\rho|\sigma}^{\delta,n}\equiv\mathrm{span}\left\{  |\phi_{y^{n}}%
\rangle:\left\vert -\frac{1}{n}\log(f_{Y^{n}}(y^{n}))+\operatorname{Tr}%
\{\rho\log\sigma\}\right\vert \leq\delta\right\}  ,\label{qtyp}%
\end{equation}
where%
\begin{align}
y^{n} &  \equiv y_{1}\cdots y_{n},\\
f_{Y^{n}}\left(  y^{n}\right)   &  \equiv\prod\limits_{i=1}^{n}f_{Y}\left(
y_{i}\right)  ,\\
|\phi_{y^{n}}\rangle &  \equiv|\phi_{y_{1}}\rangle\otimes\cdots\otimes
|\phi_{y_{n}}\rangle.
\end{align}
We will overload the notation $T_{\rho|\sigma}^{\delta,n}$ to refer also to
the following classical typical set:%
\begin{equation}
T_{\rho|\sigma}^{\delta,n}\equiv\left\{  y^{n}:\left\vert -\frac{1}{n}%
\log(f_{Y^{n}}(y^{n}))+\operatorname{Tr}\{\rho\log\sigma\}\right\vert
\leq\delta\right\}  ,
\end{equation}
with it being clear from the context whether the relative typical subspace or
set is being employed.

Let the projection operator corresponding to the relative typical subspace
$T_{\rho|\sigma}^{\delta,n}$ be called $\Pi_{\rho|\sigma,\delta}^{n}$.
Consider that%
\begin{align}
\operatorname{Tr}\{\rho\log\sigma\} &  =\operatorname{Tr}\left\{  \rho
\log\left(  \sum_{y}f_{Y}(y)|\phi_{y}\rangle\langle\phi_{y}|\right)  \right\}
\\
&  =\sum_{y}\left\langle \phi_{y}\right\vert \rho\left\vert \phi
_{y}\right\rangle \log f_{Y}(y).
\end{align}
Defining%
\begin{equation}
p_{\widetilde{Y}}\left(  y\right)  \equiv\left\langle \phi_{y}\right\vert
\rho\left\vert \phi_{y}\right\rangle ,
\end{equation}
we can then write%
\begin{align}
\operatorname{Tr}\{\rho\log\sigma\} &  =\sum_{y}p_{\widetilde{Y}}\left(
y\right)  \log f_{Y}(y)\\
&  =\mathbb{E}_{\widetilde{Y}}\left\{  \log f_{Y}(\widetilde{Y})\right\}  .
\end{align}
With this in mind, we can now calculate%
\begin{align}
\operatorname{Tr}\left\{  \Pi_{\rho|\sigma,\delta}^{n}\rho^{\otimes
n}\right\}   &  =\sum_{y^{n}\in T_{\rho|\sigma}^{\delta,n}}\langle\phi_{y^{n}%
}|\rho^{\otimes n}|\phi_{y^{n}}\rangle\\
&  =\sum_{y^{n}\in T_{\rho|\sigma}^{\delta,n}}p_{\widetilde{Y}^{n}}\left(
y^{n}\right)  \\
&  =\Pr_{\widetilde{Y}^{n}}\left\{  \widetilde{Y}^{n}\in T_{\rho|\sigma
}^{\delta,n}\right\}  .
\end{align}
Based on the above reductions, and due to the notion of typicality with
respect to the subspace $T_{\rho|\sigma}^{\delta,n}$ defined in (\ref{qtyp}),
it follows from the law of large numbers that, for a given small real number
$\varepsilon\in(0,1)$, and a sufficiently large value of $n$, $\mathrm{Tr}%
\{\Pi_{\rho|\sigma,\delta}^{n}\rho^{\otimes n}\}\geq1-\varepsilon$. In fact,
the convergence $\lim_{n\rightarrow\infty}\operatorname{Tr}\{\Pi_{\rho
|\sigma,\delta}^{n}\rho^{\otimes n}\}=1$ can be taken exponentially fast in
$n$ for a constant $\delta$ by employing the Hoeffding inequality \cite{H63}.

\section{Remainder term for monotonicity of relative entropy with respect to partial
trace}

\begin{theorem}
\label{thm:mono-PT}Let $\rho_{AB}$ be a density operator, $\sigma_{AB}$ be a
positive semi-definite operator, both acting on a finite-dimensional
tensor-product Hilbert space $\mathcal{H}_{A}\otimes\mathcal{H}_{B}$, such
that $\operatorname{supp}(\rho_{AB})\subseteq\operatorname{supp}(\sigma_{AB}%
)$, $\sigma_{B}\equiv\operatorname{Tr}_{A}\left\{  \sigma_{AB}\right\}  $ is
positive definite, and $\rho_{B}\equiv\operatorname{Tr}_{A}\left\{  \rho
_{AB}\right\}  $. Then the following inequality refines monotonicity of
relative entropy with respect to partial trace:%
\begin{equation}
D\left(  \rho_{AB}\Vert\sigma_{AB}\right)  -D\left(  \rho_{B}\Vert\sigma
_{B}\right)  \geq-\log F\left(  \rho_{AB},\left(  \mathcal{V}_{AB}%
\circ\mathcal{R}_{B\rightarrow AB}^{P}\circ\mathcal{U}_{B}\right)  \left(
\rho_{B}\right)  \right)  ,\label{eq:mono-PT}%
\end{equation}
for unitary channels $\mathcal{U}_{B}$ and $\mathcal{V}_{AB}$ defined in terms
of some unitary operators $U_{B}$ and $V_{AB}$ as%
\begin{align}
\mathcal{U}_{B}\left(  \cdot\right)   &  \equiv U_{B}\left(  \cdot\right)
U_{B}^{\dag},\\
\mathcal{V}_{AB}\left(  \cdot\right)   &  \equiv V_{AB}\left(  \cdot\right)
V_{AB}^{\dag},
\end{align}
and with $\mathcal{R}_{B\rightarrow AB}^{P}$ the CPTP\ Petz recovery map:%
\begin{equation}
\mathcal{R}_{B\rightarrow AB}^{P}\left(  \cdot\right)  \equiv\sigma_{AB}%
^{1/2}\sigma_{B}^{-1/2}\left(  \cdot\right)  \sigma_{B}^{-1/2}\sigma
_{AB}^{1/2}.
\end{equation}

\end{theorem}

\begin{proof}
Our proof of Theorem~\ref{thm:mono-PT}
proceeds very similarly to the proof of \cite[Theorem~5.1]{FR14}, with only a
few modifications. We give a full proof for completeness. Our proof makes use of Lemmas~2.3, 4.2, B.2, B.6, and B.7 from \cite{FR14}. For convenience of the reader, we recall these statements in Appendix~\ref{app:FR-lemmas}.

The expression on the left-hand side of \eqref{eq:mono-PT} is equivalent to%
\begin{equation}
-H\left(  A|B\right)  _{\rho}-\operatorname{Tr}\left\{  \rho_{AB}\log
\sigma_{AB}\right\}  +\operatorname{Tr}\left\{  \rho_{B}\log\sigma
_{B}\right\}  ,
\end{equation}
where $H\left(  A|B\right)  _{\rho}\equiv H\left(  AB\right)  _{\rho}-H\left(
B\right)  _{\rho}$ is the conditional entropy and the entropy is defined as
$H\left(  \omega\right)  \equiv-\operatorname{Tr}\left\{  \omega\log
\omega\right\}  $. So we need the relative typical projectors $\Pi_{\rho
_{AB}|\sigma_{AB},\delta}^{n}$ and $\Pi_{\rho_{B}|\sigma_{B},\delta}^{n}$
defined in Section~\ref{sec:stein-review}. Abbreviate these as $\Pi_{AB}^{n}$
and $\Pi_{B}^{n}$, respectively.

We begin by defining%
\begin{equation}
\mathcal{W}_{n}\left(  X_{A^{n}B^{n}}\right)  \equiv\Pi_{AB}^{n}\Pi_{B}%
^{n}X_{A^{n}B^{n}}\Pi_{B}^{n}\Pi_{AB}^{n}.
\end{equation}
We employ the shorthand $\mathcal{W}_{n}\left(  X_{B^{n}}\right)
\equiv\mathcal{W}_{n}\left(  I_{A}^{\otimes n}\otimes X_{B^{n}}\right)  $
throughout. Consider from the gentle measurement lemma \cite{W99}, properties
of the trace norm,\ and relative typicality that%
\begin{align}
\operatorname{Tr}\left\{  \mathcal{W}_{n}\left(  \rho_{AB}^{\otimes n}\right)
\right\}   &  =\operatorname{Tr}\left\{  \Pi_{AB}^{n}\Pi_{B}^{n}\rho
_{AB}^{\otimes n}\Pi_{B}^{n}\right\}  \\
&  \geq\operatorname{Tr}\left\{  \Pi_{AB}^{n}\rho_{AB}^{\otimes n}\right\}
-\left\Vert \Pi_{B}^{n}\rho_{AB}^{\otimes n}\Pi_{B}^{n}-\rho_{AB}^{\otimes
n}\right\Vert _{1}\\
&  \geq1-\eta,
\end{align}
where $\eta$ is an arbitrarily small positive number for sufficiently large
$n$. So we apply \cite[Lemma~2.3]{FR14} to find that%
\begin{align}
D\left(  \mathcal{W}_{n}\left(  \rho_{AB}^{\otimes n}\right)  \Vert
\mathcal{W}_{n}\left(  \rho_{B}^{\otimes n}\right)  \right)   &  \leq n\left(
D\left(  \rho_{AB}\Vert I_{A}\otimes\rho_{B}\right)  +\frac{\delta}{2}\right)
\\
&  =n\left(  -H\left(  A|B\right)  _{\rho}+\frac{\delta}{2}\right)  ,
\end{align}
where the above inequality holds for sufficiently large $n$. A well-known
relation between the root fidelity $\sqrt{F}\left(  \omega,\tau\right)
\equiv\left\Vert \sqrt{\omega}\sqrt{\tau}\right\Vert _{1}$\ and relative
entropy \cite[Lemma~B.2]{FR14}\ then gives that%
\begin{equation}
\frac{1}{\operatorname{Tr}\left\{  \mathcal{W}_{n}\left(  \rho_{AB}^{\otimes
n}\right)  \right\}  }\sqrt{F}\left(  \mathcal{W}_{n}\left(  \rho
_{AB}^{\otimes n}\right)  ,\mathcal{W}_{n}\left(  \rho_{B}^{\otimes n}\right)
\right)  \geq2^{\frac{1}{2}n\left(  H\left(  A|B\right)  _{\rho}-\frac{\delta
}{2}\right)  }.
\end{equation}
Use \cite[Lemma~B.6]{FR14} to remove the projector $\Pi_{AB}^{n}$ from the
second argument, so that%
\begin{equation}
\frac{1}{\operatorname{Tr}\left\{  \mathcal{W}_{n}\left(  \rho_{AB}^{\otimes
n}\right)  \right\}  }\sqrt{F}\left(  \mathcal{W}_{n}\left(  \rho
_{AB}^{\otimes n}\right)  ,\Pi_{B}^{n}\rho_{B}^{\otimes n}\Pi_{B}^{n}\right)
\geq2^{\frac{1}{2}n\left(  H\left(  A|B\right)  _{\rho}-\frac{\delta}%
{2}\right)  },
\end{equation}
and the trace term can be eliminated at the expense of decreasing the exponent
by a constant times~$n$:%
\begin{equation}
\sqrt{F}\left(  \mathcal{W}_{n}\left(  \rho_{AB}^{\otimes n}\right)  ,\Pi
_{B}^{n}\rho_{B}^{\otimes n}\Pi_{B}^{n}\right)  \geq2^{\frac{1}{2}n\left(
H\left(  A|B\right)  _{\rho}-\delta\right)  }.
\end{equation}
Let an eigendecomposition of $\sigma_{B}^{\otimes n}$ be given as%
\begin{equation}
\sigma_{B}^{\otimes n}=\sum_{s\in S_{n}}s\Pi_{s},\label{eq:method-types-1}%
\end{equation}
where $S_{n}$ is the set of eigenvalues of $\sigma_{B}^{\otimes n}$. By
defining%
\begin{equation}
S_{n,\delta}\equiv\left\{  s\in S_{n}:\left\vert -\frac{1}{n}\log
(s)+\mathrm{Tr}\{\rho_{B}\log\sigma_{B}\}\right\vert \leq\delta\right\}  ,
\end{equation}
we see from (\ref{qtyp}) and the definition of $\Pi_{B}^{n}$ that%
\begin{equation}
\Pi_{B}^{n}=\sum_{s\in S_{n,\delta}}\Pi_{s}.\label{eq:method-types-4}%
\end{equation}
Furthermore, it follows from a trivial combinatorial consideration that
$\left\vert S_{n,\delta}\right\vert \leq\operatorname{poly}\left(  n\right)
$. Then consider that $\sum_{s}\Pi_{s}=I$ and apply \cite[Lemma~B.7]{FR14} to
get%
\begin{align}
\sqrt{F}\left(  \mathcal{W}_{n}\left(  \rho_{AB}^{\otimes n}\right)  ,\Pi
_{B}^{n}\rho_{B}^{\otimes n}\Pi_{B}^{n}\right)   &  \leq\sum_{s\in S_{n}}%
\sqrt{F}\left(  \mathcal{W}_{n}\left(  \rho_{AB}^{\otimes n}\right)  ,\Pi
_{s}\Pi_{B}^{n}\rho_{B}^{\otimes n}\Pi_{B}^{n}\Pi_{s}\right)  \\
&  =\sum_{s\in S_{n,\delta}}\sqrt{F}\left(  \mathcal{W}_{n}\left(  \rho
_{AB}^{\otimes n}\right)  ,\Pi_{s}\rho_{B}^{\otimes n}\Pi_{s}\right)
\label{eq:types-orthogonal}\\
&  \leq\left\vert S_{n,\delta}\right\vert \max_{s\in S_{n,\delta}}\sqrt
{F}\left(  \mathcal{W}_{n}\left(  \rho_{AB}^{\otimes n}\right)  ,\Pi_{s}%
\rho_{B}^{\otimes n}\Pi_{s}\right)  ,
\end{align}
where (\ref{eq:types-orthogonal}) follows because $\Pi_{s}\Pi_{B}^{n}=\Pi_{s}$
if $s\in S_{n,\delta}$ and it is equal to zero otherwise. So we find that
there exists an $s$ such that%
\begin{equation}
\sqrt{F}\left(  \mathcal{W}_{n}\left(  \rho_{AB}^{\otimes n}\right)  ,\Pi
_{B}^{n}\rho_{B}^{\otimes n}\Pi_{B}^{n}\right)  \leq\text{poly}\left(
n\right)  \sqrt{F}\left(  \mathcal{W}_{n}\left(  \rho_{AB}^{\otimes n}\right)
,\Pi_{s}\rho_{B}^{\otimes n}\Pi_{s}\right)  .
\end{equation}
From the definition of $\Pi_{s}$ we can write%
\begin{equation}
\Pi_{s}=\sqrt{s}\left(  \sigma_{B}^{-1/2}\right)  ^{\otimes n}\Pi_{s}.
\end{equation}
From the definition of $S_{n,\delta}$, we have that%
\begin{equation}
\sqrt{s}\leq2^{\frac{1}{2}n\left[  \operatorname{Tr}\left\{  \rho_{B}%
\log\sigma_{B}\right\}  +\delta\right]  },
\end{equation}
giving that%
\begin{align}
&  \sqrt{F}\left(  \mathcal{W}_{n}\left(  \rho_{AB}^{\otimes n}\right)
,\Pi_{s}\rho_{B}^{\otimes n}\Pi_{s}\right)  \nonumber\\
&  =\sqrt{s}\sqrt{F}\left(  \mathcal{W}_{n}\left(  \rho_{AB}^{\otimes
n}\right)  ,\left(  \sigma_{B}^{-1/2}\right)  ^{\otimes n}\Pi_{s}\left(
\rho_{B}^{\otimes n}\right)  \Pi_{s}\left(  \sigma_{B}^{-1/2}\right)
^{\otimes n}\right)  \\
&  \leq2^{\frac{1}{2}n\left[  \operatorname{Tr}\left\{  \rho_{B}\log\sigma
_{B}\right\}  +\delta\right]  }\sqrt{F}\left(  \mathcal{W}_{n}\left(
\rho_{AB}^{\otimes n}\right)  ,\left(  \sigma_{B}^{-1/2}\right)  ^{\otimes
n}\Pi_{s}\rho_{B}^{\otimes n}\Pi_{s}\left(  \sigma_{B}^{-1/2}\right)
^{\otimes n}\right)  \\
&  =2^{\frac{1}{2}n\left[  \operatorname{Tr}\left\{  \rho_{B}\log\sigma
_{B}\right\}  +\delta\right]  }\sqrt{F}\left(  \Pi_{s}\left(  \sigma
_{B}^{-1/2}\right)  ^{\otimes n}\mathcal{W}_{n}\left(  \rho_{AB}^{\otimes
n}\right)  \left(  \sigma_{B}^{-1/2}\right)  ^{\otimes n}\Pi_{s},\rho
_{B}^{\otimes n}\right)  ,
\end{align}
where the last equality is from \cite[Lemma~B.6]{FR14}. Now, by \cite[Lemma
4.2]{FR14}, there exists a unitary $U_{B}$\ such that\footnote{Note that the
unitary $U_{B}$ depends on $n$, but we suppress this in the notation for
simplicity.}%
\begin{align}
&  \sqrt{F}\left(  \Pi_{s}\left(  \sigma_{B}^{-1/2}\right)  ^{\otimes
n}\mathcal{W}_{n}\left(  \rho_{AB}^{\otimes n}\right)  \left(  \sigma
_{B}^{-1/2}\right)  ^{\otimes n}\Pi_{s},\rho_{B}^{\otimes n}\right)
\nonumber\\
&  \leq\text{poly}\left(  n\right)  \sqrt{F}\left(  \left(  \sigma_{B}%
^{-1/2}\right)  ^{\otimes n}\mathcal{W}_{n}\left(  \rho_{AB}^{\otimes
n}\right)  \left(  \sigma_{B}^{-1/2}\right)  ^{\otimes n},U_{B}^{\otimes
n}\rho_{B}^{\otimes n}\left(  U_{B}^{\otimes n}\right)  ^{\dag}\right)  \\
&  =\text{poly}\left(  n\right)  \sqrt{F}\left(  \mathcal{W}_{n}\left(
\rho_{AB}^{\otimes n}\right)  ,\left(  \sigma_{B}^{-1/2}\right)  ^{\otimes
n}U_{B}^{\otimes n}\rho_{B}^{\otimes n}\left(  U_{B}^{\otimes n}\right)
^{\dag}\left(  \sigma_{B}^{-1/2}\right)  ^{\otimes n}\right)  .
\end{align}
The equality above follows by applying \cite[Lemma~B.6]{FR14}. Combining
everything up until now, we get%
\begin{multline}
2^{\frac{1}{2}n\left(  H\left(  A|B\right)  _{\rho}-\operatorname{Tr}\left\{
\rho_{B}\log\sigma_{B}\right\}  -2\delta\right)  }\\
\leq\text{poly}\left(  n\right)  \sqrt{F}\left(  \Pi_{AB}^{n}\Pi_{B}^{n}%
\rho_{AB}^{\otimes n}\Pi_{B}^{n}\Pi_{AB}^{n},\left(  \sigma_{B}^{-1/2}\right)
^{\otimes n}U_{B}^{\otimes n}\rho_{B}^{\otimes n}\left(  U_{B}^{\otimes
n}\right)  ^{\dag}\left(  \sigma_{B}^{-1/2}\right)  ^{\otimes n}\right)  .
\end{multline}

Let an eigendecomposition of $\sigma_{AB}^{\otimes n}$ be given as%
\begin{equation}
\sigma_{AB}^{\otimes n}=\sum_{p\in P_{n}}p\Pi_{p},
\end{equation}
and%
\begin{equation}
\Pi_{AB}^{n}=\sum_{p\in P_{n,\delta}}\Pi_{p},
\end{equation}
where these developments follow the same reasoning as (\ref{eq:method-types-1}%
)-(\ref{eq:method-types-4}). Now we continue with the fact that $\sum_{p\in
P_{n}}\Pi_{p}=I$ and \cite[Lemma~B.7]{FR14} to get that%
\begin{align}
&  \sqrt{F}\left(  \Pi_{AB}^{n}\Pi_{B}^{n}\rho_{AB}^{\otimes n}\Pi_{B}^{n}%
\Pi_{AB}^{n},\left(  \sigma_{B}^{-1/2}\right)  ^{\otimes n}U_{B}^{\otimes
n}\rho_{B}^{\otimes n}\left(  U_{B}^{\otimes n}\right)  ^{\dag}\left(
\sigma_{B}^{-1/2}\right)  ^{\otimes n}\right)  \nonumber\\
&  \leq\sum_{p\in P_{n}}\sqrt{F}\left(  \Pi_{p}\Pi_{AB}^{n}\Pi_{B}^{n}%
\rho_{AB}^{\otimes n}\Pi_{B}^{n}\Pi_{AB}^{n}\Pi_{p},\left(  \sigma_{B}%
^{-1/2}\right)  ^{\otimes n}U_{B}^{\otimes n}\rho_{B}^{\otimes n}\left(
U_{B}^{\otimes n}\right)  ^{\dag}\left(  \sigma_{B}^{-1/2}\right)  ^{\otimes
n}\right)  \\
&  =\sum_{p\in P_{n,\delta}}\sqrt{F}\left(  \Pi_{p}\Pi_{B}^{n}\rho
_{AB}^{\otimes n}\Pi_{B}^{n}\Pi_{p},\left(  \sigma_{B}^{-1/2}\right)
^{\otimes n}U_{B}^{\otimes n}\rho_{B}^{\otimes n}\left(  U_{B}^{\otimes
n}\right)  ^{\dag}\left(  \sigma_{B}^{-1/2}\right)  ^{\otimes n}\right)  \\
&  \leq\left\vert P_{n,\delta}\right\vert \max_{p\in P_{n,\delta}}\sqrt
{F}\left(  \Pi_{p}\Pi_{B}^{n}\rho_{AB}^{\otimes n}\Pi_{B}^{n}\Pi_{p},\left(
\sigma_{B}^{-1/2}\right)  ^{\otimes n}U_{B}^{\otimes n}\rho_{B}^{\otimes
n}\left(  U_{B}^{\otimes n}\right)  ^{\dag}\left(  \sigma_{B}^{-1/2}\right)
^{\otimes n}\right)  .
\end{align}
Then there exists a $p$ such that%
\begin{multline}
\sqrt{F}\left(  \Pi_{AB}^{n}\Pi_{B}^{n}\rho_{AB}^{\otimes n}\Pi_{B}^{n}%
\Pi_{AB}^{n},\left(  \sigma_{B}^{-1/2}\right)  ^{\otimes n}U_{B}^{\otimes
n}\rho_{B}^{\otimes n}\left(  U_{B}^{\otimes n}\right)  ^{\dag}\left(
\sigma_{B}^{-1/2}\right)  ^{\otimes n}\right)  \\
\leq\text{poly}\left(  n\right)  \sqrt{F}\left(  \Pi_{p}\Pi_{B}^{n}\rho
_{AB}^{\otimes n}\Pi_{B}^{n}\Pi_{p},\left(  \sigma_{B}^{-1/2}\right)
^{\otimes n}U_{B}^{\otimes n}\rho_{B}^{\otimes n}\left(  U_{B}^{\otimes
n}\right)  ^{\dag}\left(  \sigma_{B}^{-1/2}\right)  ^{\otimes n}\right)  .
\end{multline}
From the definition of $\Pi_{p}$ we have that%
\begin{equation}
\Pi_{p}=\frac{1}{\sqrt{p}}\left(  \sigma_{AB}^{1/2}\right)  ^{\otimes n}%
\Pi_{p},
\end{equation}
with $\sqrt{p}\geq2^{\frac{1}{2}n\left[  \operatorname{Tr}\left\{  \rho
_{AB}\log\sigma_{AB}\right\}  -\delta\right]  }$. Then by defining
$K\equiv2^{\frac{1}{2}n\left[  \operatorname{Tr}\left\{  \rho_{AB}\log
\sigma_{AB}\right\}  -\delta\right]  }/\sqrt{p}$, we have that%
\begin{align}
&  2^{\frac{1}{2}n\left[  \operatorname{Tr}\left\{  \rho_{AB}\log\sigma
_{AB}\right\}  -\delta\right]  }\sqrt{F}\left(  \Pi_{p}\Pi_{B}^{n}\rho
_{AB}^{\otimes n}\Pi_{B}^{n}\Pi_{p},\left(  \sigma_{B}^{-1/2}\right)
^{\otimes n}U_{B}^{\otimes n}\rho_{B}^{\otimes n}\left(  U_{B}^{\otimes
n}\right)  ^{\dag}\left(  \sigma_{B}^{-1/2}\right)  ^{\otimes n}\right)
\nonumber\\
&  =K\ \sqrt{F}\left(  \left(  \sigma_{AB}^{1/2}\right)  ^{\otimes n}\Pi
_{p}\Pi_{B}^{n}\rho_{AB}^{\otimes n}\Pi_{B}^{n}\Pi_{p}\left(  \sigma
_{AB}^{1/2}\right)  ^{\otimes n},\left(  \sigma_{B}^{-1/2}\right)  ^{\otimes
n}U_{B}^{\otimes n}\rho_{B}^{\otimes n}\left(  U_{B}^{\otimes n}\right)
^{\dag}\left(  \sigma_{B}^{-1/2}\right)  ^{\otimes n}\right)  \\
&  \leq\sqrt{F}\left(  \left(  \sigma_{AB}^{1/2}\right)  ^{\otimes n}\Pi
_{p}\Pi_{B}^{n}\rho_{AB}^{\otimes n}\Pi_{B}^{n}\Pi_{p}\left(  \sigma
_{AB}^{1/2}\right)  ^{\otimes n},\left(  \sigma_{B}^{-1/2}\right)  ^{\otimes
n}U_{B}^{\otimes n}\rho_{B}^{\otimes n}\left(  U_{B}^{\otimes n}\right)
^{\dag}\left(  \sigma_{B}^{-1/2}\right)  ^{\otimes n}\right)  \\
&  =\sqrt{F}\left(  \Pi_{p}\Pi_{B}^{n}\rho_{AB}^{\otimes n}\Pi_{B}^{n}\Pi
_{p},\left(  \sigma_{AB}^{1/2}\right)  ^{\otimes n}\left(  \sigma_{B}%
^{-1/2}\right)  ^{\otimes n}U_{B}^{\otimes n}\rho_{B}^{\otimes n}\left(
U_{B}^{\otimes n}\right)  ^{\dag}\left(  \sigma_{B}^{-1/2}\right)  ^{\otimes
n}\left(  \sigma_{AB}^{1/2}\right)  ^{\otimes n}\right)  .
\end{align}
Now by \cite[Lemma 4.2]{FR14}, there exists a unitary $V_{AB}$ such
that\footnote{Note that the unitary $V_{AB}$ depends on $n$, but we suppress
this in the notation for simplicity.}%
\begin{multline}
\sqrt{F}\left(  \Pi_{p}\Pi_{B}^{n}\rho_{AB}^{\otimes n}\Pi_{B}^{n}\Pi
_{p},\left(  \sigma_{AB}^{1/2}\right)  ^{\otimes n}\left(  \sigma_{B}%
^{-1/2}\right)  ^{\otimes n}U_{B}^{\otimes n}\rho_{B}^{\otimes n}\left(
U_{B}^{\otimes n}\right)  ^{\dag}\left(  \sigma_{B}^{-1/2}\right)  ^{\otimes
n}\left(  \sigma_{AB}^{1/2}\right)  ^{\otimes n}\right)  \leq
\label{eq:use-for-support}\\
\text{poly}\left(  n\right)  \sqrt{F}\left(  \rho_{AB}^{\otimes n}%
,V_{AB}^{\otimes n}\left(  \sigma_{AB}^{1/2}\right)  ^{\otimes n}\left(
\sigma_{B}^{-1/2}\right)  ^{\otimes n}U_{B}^{\otimes n}\rho_{B}^{\otimes
n}\left(  U_{B}^{\otimes n}\right)  ^{\dag}\left(  \sigma_{B}^{-1/2}\right)
^{\otimes n}\left(  \sigma_{AB}^{1/2}\right)  ^{\otimes n}\left(
V_{AB}^{\otimes n}\right)  ^{\dag}\right)  .
\end{multline}
Putting everything together, we get that%
\begin{align}
&  2^{\frac{1}{2}n\left(  H\left(  A|B\right)  _{\rho}-\operatorname{Tr}%
\left\{  \rho_{B}\log\sigma_{B}\right\}  +\operatorname{Tr}\left\{  \rho
_{AB}\log\sigma_{AB}\right\}  -3\delta\right)  }\nonumber\\
&  \leq\text{poly}\left(  n\right)  \sqrt{F}\left(  \rho_{AB}^{\otimes
n},V_{AB}^{\otimes n}\left(  \sigma_{AB}^{1/2}\right)  ^{\otimes n}\left(
\sigma_{B}^{-1/2}\right)  ^{\otimes n}U_{B}^{\otimes n}\rho_{B}^{\otimes
n}\left(  U_{B}^{\otimes n}\right)  ^{\dag}\left(  \sigma_{B}^{-1/2}\right)
^{\otimes n}\left(  \sigma_{AB}^{1/2}\right)  ^{\otimes n}\left(
V_{AB}^{\otimes n}\right)  ^{\dag}\right)  \ \\
&  =\text{poly}\left(  n\right)  \left[  F\left(  \rho_{AB},V_{AB}\sigma
_{AB}^{1/2}\sigma_{B}^{-1/2}U_{B}\rho_{B}U_{B}^{\dag}\sigma_{B}^{-1/2}%
\sigma_{AB}^{1/2}V_{AB}^{\dag}\right)  \right]  ^{n}\\
&  \leq\text{poly}\left(  n\right)  \left[  \max_{U_{B},V_{AB}}F\left(
\rho_{AB},V_{AB}\sigma_{AB}^{1/2}\sigma_{B}^{-1/2}U_{B}\rho_{B}U_{B}^{\dag
}\sigma_{B}^{-1/2}\sigma_{AB}^{1/2}V_{AB}^{\dag}\right)  \right]
^{n}.\label{eq:opt-over-U}%
\end{align}
The equality follows because the fidelity is multiplicative with respect to tensor
products. In the last line above, we take a maximization over all unitaries in
order to remove the dependence of the unitaries on $n$. Taking the
$n^{\text{th}}$ root of the last line above, we find that there exists a
$V_{AB}$ and $U_{B}$ such that%
\begin{multline}
2^{\frac{1}{2}\left(  H\left(  A|B\right)  _{\rho}-\operatorname{Tr}\left\{
\rho_{B}\log\sigma_{B}\right\}  +\operatorname{Tr}\left\{  \rho_{AB}\log
\sigma_{AB}\right\}  -3\delta\right)  }\\
\leq\sqrt[n]{\text{poly}\left(  n\right)  }\sqrt{F}\left(  \rho_{AB}%
,V_{AB}\sigma_{AB}^{1/2}\sigma_{B}^{-1/2}U_{B}\rho_{B}U_{B}^{\dag}\sigma
_{B}^{-1/2}\sigma_{AB}^{1/2}V_{AB}^{\dag}\right)  .
\end{multline}
By taking the limit as $n$ becomes large, using the fact that
\begin{equation}  -\left[H\left(  A|B\right)  _{\rho}-\operatorname{Tr}\left\{
\rho_{B}\log\sigma_{B}\right\}  +\operatorname{Tr}\left\{  \rho_{AB}\log
\sigma_{AB}\right\}\right] = D\left(  \rho_{AB}\Vert\sigma_{AB}\right)  -D\left(  \rho_{B}\Vert\sigma
_{B}\right),\end{equation}
and noting that $\delta>0$ was arbitrary, this finally yields the desired
inequality%
\begin{equation}
D\left(  \rho_{AB}\Vert\sigma_{AB}\right)  -D\left(  \rho_{B}\Vert\sigma
_{B}\right)  \geq-\log F\left(  \rho_{AB},V_{AB}\sigma_{AB}^{1/2}\sigma
_{B}^{-1/2}U_{B}\rho_{B}U_{B}^{\dag}\sigma_{B}^{-1/2}\sigma_{AB}^{1/2}%
V_{AB}^{\dag}\right)  .
\end{equation}

\end{proof}

\begin{remark}
\label{rmk:partial_trace} Suppose in Theorem~\ref{thm:mono-PT}\ that
$\sigma_{AB}$ is a density operator. It remains open to quantify the
performance of the rotated Petz recovery map $\mathcal{V}_{AB}\circ
\mathcal{R}_{B\rightarrow AB}^{P}\circ\mathcal{U}_{B}$ on the reduced state
$\sigma_{B}$. In particular, if the unitary channels $\mathcal{U}_{B}$ and
$\mathcal{V}_{AB}$ were not necessary (with each instead being equal to the
identity channel), then it would be possible to do so. This form of the
recovery map was previously conjectured in \cite[Consequence~27]{SBW14} in
terms of the following inequality:%
\begin{equation}
D\left(  \rho_{AB}\Vert\sigma_{AB}\right)  -D\left(  \rho_{B}\Vert\sigma
_{B}\right)  \geq-\log F\left(  \rho_{AB},\mathcal{R}_{B\rightarrow AB}%
^{P}\left(  \rho_{B}\right)  \right)  . \label{eq:conjecture}%
\end{equation}
If this conjecture is true, then one could perform the Petz recovery map on
system $B$ and be guaranteed a perfect recovery of $\sigma_{AB}$ if the state
of $B$ is $\sigma_{B}$, while having a performance limited by
\eqref{eq:conjecture} if the state of $B$ is $\rho_{B}$. By a modification of
the proof of Theorem~\ref{thm:mono-PT}, one can also establish the following
lower bound:%
\begin{equation}
D\left(  \rho_{AB}\Vert\sigma_{AB}\right)  -D\left(  \rho_{B}\Vert\sigma
_{B}\right)  \geq-\log F\left(  \rho_{AB},\sigma_{AB}^{1/2}\bar{V}_{AB}\bar
{U}_{B}\sigma_{B}^{-1/2}\rho_{B}\sigma_{B}^{-1/2}\bar{U}_{B}^{\dag}\bar
{V}_{AB}^{\dag}\sigma_{AB}^{1/2}\right)  , \label{eq:alt-bound}%
\end{equation}
for some unitaries $\bar{U}_{B}$ and $\bar{V}_{AB}$. The completely positive
map $\sigma_{AB}^{1/2}\bar{V}_{AB}\bar{U}_{B}\sigma_{B}^{-1/2}\left(
\cdot\right)  \sigma_{B}^{-1/2}\bar{U}_{B}^{\dag}\bar{V}_{AB}^{\dag}%
\sigma_{AB}^{1/2}$ recovers $\sigma_{AB}$ perfectly from $\sigma_{B}$, while
having a performance limited by \eqref{eq:alt-bound} when recovering
$\rho_{AB}$ from $\rho_{B}$. It is however unclear whether this
map is trace preserving.
\end{remark}

\section{Remainder term for monotonicity of relative entropy}

\begin{corollary}
\label{cor:mono}Let $\rho_{S}$ be a density operator and $\sigma_{S}$ be a
positive semi-definite operator, both acting on a Hilbert space $\mathcal{H}%
_{S}$ and such that $\operatorname{supp}(\rho_{S})\subseteq\operatorname{supp}%
(\sigma_{S})$. Let $\mathcal{N}_{S\rightarrow B}$ be a CPTP\ map taking
density operators acting on $\mathcal{H}_{S}$ to density operators acting on
$\mathcal{H}_{B}$ and such that $\mathcal{N}_{S\rightarrow B}\left(
\sigma_{S}\right)  $ is a positive definite operator. Then the following
inequality refines monotonicity of relative entropy:%
\begin{equation}
D\left(  \rho_{S}\Vert\sigma_{S}\right)  -D\left(  \mathcal{N}_{S\rightarrow
B}\left(  \rho_{S}\right)  \Vert\mathcal{N}_{S\rightarrow B}\left(  \sigma
_{S}\right)  \right)  \geq-\log F\left(  \rho_{S},\left(  \mathcal{V}_{S}%
\circ\mathcal{R}_{\sigma,\mathcal{N}}^{P}\circ\mathcal{U}_{B}\right)  \left(
\mathcal{N}_{S\rightarrow B}\left(  \rho_{S}\right)  \right)  \right)
,\label{eq:mono-rel-ent-remainder}%
\end{equation}
for unitary channels $\mathcal{U}_{B}$ and $\mathcal{V}_{S}$ defined in terms
of some unitary operators $U_{B}$ and $V_{S}$ as%
\begin{align}
\mathcal{U}_{B}\left(  \cdot\right)   &  \equiv U_{B}\left(  \cdot\right)
U_{B}^{\dag},\\
\mathcal{V}_{S}\left(  \cdot\right)   &  \equiv V_{S}\left(  \cdot\right)
V_{S}^{\dag},
\end{align}
and with $\mathcal{R}_{\sigma,\mathcal{N}}^{P}$ the CPTP\ Petz recovery map:%
\begin{equation}
\mathcal{R}_{\sigma,\mathcal{N}}^{P}\left(  \cdot\right)  \equiv\sigma
_{S}^{1/2}\mathcal{N}^{\dag}\left[  \left(  \mathcal{N}_{S\rightarrow
B}\left(  \sigma_{S}\right)  \right)  ^{-1/2}\left(  \cdot\right)  \left(
\mathcal{N}_{S\rightarrow B}\left(  \sigma_{S}\right)  \right)  ^{-1/2}%
\right]  \sigma_{S}^{1/2},
\end{equation}
where $\mathcal{N}^{\dag}$ is the adjoint of $\mathcal{N}_{S\rightarrow B}$.
\end{corollary}

\begin{proof}
[Proof of Theorem~\ref{cor:mono}]We begin by recalling that any quantum
channel can be realized by tensoring in an ancilla system prepared in a
fiducial state, acting with a unitary on the input and ancilla, and then
performing a partial trace \cite{S55}. That is, for any channel $\mathcal{N}%
_{S\rightarrow B}$, there exists a unitary $W_{SE^{\prime}\rightarrow BE}$
with input systems $SE^{\prime}$ and output systems $BE$ such that%
\begin{equation}
\mathcal{N}_{S\rightarrow B}\left(  \rho_{S}\right)  =\operatorname{Tr}%
_{E}\left\{  W_{SE^{\prime}\rightarrow BE}\left(  \rho_{S}\otimes\left\vert
0\right\rangle \left\langle 0\right\vert _{E^{\prime}}\right)  W_{SE^{\prime
}\rightarrow BE}^{\dag}\right\}  .\label{eq:unitary-extension}%
\end{equation}
For simplicity, we abbreviate the unitary $W_{SE^{\prime}\rightarrow BE}$ as
$W$ in what follows. Let $\rho_{BE}$ and $\sigma_{BE}$ be defined as%
\begin{align}
\rho_{BE} &  \equiv W\left(  \rho_{S}\otimes\left\vert 0\right\rangle
\left\langle 0\right\vert _{E^{\prime}}\right)  W^{\dag},\\
\sigma_{BE} &  \equiv W\left(  \sigma_{S}\otimes\left\vert 0\right\rangle
\left\langle 0\right\vert _{E^{\prime}}\right)  W^{\dag},
\end{align}
so that%
\begin{equation}
\mathcal{N}_{S\rightarrow B}\left(  \rho_{S}\right)  =\rho_{B}%
,\ \ \ \ \ \ \ \ \mathcal{N}_{S\rightarrow B}\left(  \sigma_{S}\right)
=\sigma_{B}.
\end{equation}
The Kraus operators of $\mathcal{N}_{S\rightarrow B}$ are given as%
\begin{align}
\mathcal{N}_{S\rightarrow B}\left(  \rho_{S}\right)   &  =\sum_{i}\left\langle
i\right\vert _{E}W\left(  \rho_{S}\otimes\left\vert 0\right\rangle
\left\langle 0\right\vert _{E^{\prime}}\right)  W^{\dag}\left\vert
i\right\rangle _{E}\\
&  =\sum_{i}\left\langle i\right\vert _{E}W\left\vert 0\right\rangle
_{E^{\prime}}\rho_{S}\left\langle 0\right\vert _{E^{\prime}}W^{\dag}\left\vert
i\right\rangle _{E},
\end{align}
so that the adjoint map is given by%
\begin{equation}
\mathcal{N}^{\dag}\left(  \omega_{B}\right)  =\sum_{i}\left\langle
0\right\vert _{E^{\prime}}W^{\dag}\left\vert i\right\rangle _{E}\omega
_{B}\left\langle i\right\vert _{E}W\left\vert 0\right\rangle _{E^{\prime}}.
\end{equation}
Furthermore, we have that%
\begin{align}
&  D\left(  \rho_{S}\Vert\sigma_{S}\right)  -D\left(  \mathcal{N}%
_{S\rightarrow B}\left(  \rho_{S}\right)  \Vert\mathcal{N}_{S\rightarrow
B}\left(  \sigma_{S}\right)  \right)  \nonumber\\
&  =D\left(  \rho_{S}\otimes\left\vert 0\right\rangle \left\langle
0\right\vert _{E^{\prime}}\Vert\sigma_{S}\otimes\left\vert 0\right\rangle
\left\langle 0\right\vert _{E^{\prime}}\right)  -D\left(  \rho_{B}\Vert
\sigma_{B}\right)  \\
&  =D\left(  W\left(  \rho_{S}\otimes\left\vert 0\right\rangle \left\langle
0\right\vert _{E^{\prime}}\right)  W^{\dag}\Vert W\left(  \sigma_{S}%
\otimes\left\vert 0\right\rangle \left\langle 0\right\vert _{E^{\prime}%
}\right)  W^{\dag}\right)  -D\left(  \rho_{B}\Vert\sigma_{B}\right)  \\
&  =D\left(  \rho_{BE}\Vert\sigma_{BE}\right)  -D\left(  \rho_{B}\Vert
\sigma_{B}\right)  .\label{eq:reduction-to-PT}%
\end{align}
Applying Theorem~\ref{thm:mono-PT}, we know that a lower bound on
(\ref{eq:reduction-to-PT}) is%
\begin{equation}
-\log F\left(  \rho_{BE},V_{BE}\sigma_{BE}^{1/2}\sigma_{B}^{-1/2}U_{B}\rho
_{B}U_{B}^{\dag}\sigma_{B}^{-1/2}\sigma_{BE}^{1/2}V_{BE}^{\dag}\right)
,\label{eq:remainder-term}%
\end{equation}
for some unitaries $V_{BE}$ and $U_{B}$. Without loss of generality, $V_{BE}$
can be assumed to be
an isometry on the image of  $W_{SE^{\prime}\rightarrow
BE}\left\vert 0\right\rangle _{E^{\prime}}$. We justify this as follows. Let
$P_{n}$ denote the support projection of $\rho_{AB}^{\otimes n}$. By
(\ref{eq:use-for-support}), since the supports of $\Pi_{b}\Pi_{B^{n}}P_{n}$,
$\rho_{AB}^{\otimes n}$, and $\left[  \sigma_{AB}^{1/2}\sigma_{B}^{-1/2}%
U_{B}\rho_{B}U_{B}^{\dag}\sigma_{B}^{-1/2}\sigma_{AB}^{1/2}\right]  ^{\otimes
n}$ are all contained in the support of $\sigma_{AB}^{\otimes n}$, one can
apply \cite[Lemma~4.2]{FR14} on the Hilbert space $\left[  \operatorname{supp}\left(
\sigma_{AB}\right)  \right]  ^{\otimes n}=\operatorname{supp}\left(  \sigma_{AB}^{\otimes
n}\right)  $ to obtain a unitary $V_{AB}$ on this space, which may be
extended to a unitary on the space $\mathcal{H}_{AB}^{\otimes n}$ in an
arbitrary way. Hence, the maximization in (\ref{eq:opt-over-U}) can be
restricted to unitaries $V_{AB}$ that are isometries on the support of
$\sigma_{AB}$. Thus, we indeed have that $V_{BE}$ is an isometry on the
support of $\sigma_{AB}$, which can be extended to an isometry on the image of
$W_{SE^{\prime}\rightarrow BE}\left\vert 0\right\rangle _{E^{\prime}}$.

Let us now unravel the term $\sigma_{BE}^{1/2}\sigma_{B}^{-1/2}U_{B}\rho
_{B}U_{B}^{\dag}\sigma_{B}^{-1/2}\sigma_{BE}^{1/2}$ in the second argument
above. Letting
\begin{equation}
\omega_{B}\equiv\left(  \mathcal{N}_{S\rightarrow B}\left(  \sigma_{S}\right)
\right)  ^{-1/2}U_{B}\mathcal{N}_{S\rightarrow B}\left(  \rho_{S}\right)
U_{B}^{\dag}\left(  \mathcal{N}_{S\rightarrow B}\left(  \sigma_{S}\right)
\right)  ^{-1/2},
\end{equation}
we then have that%
\begin{align}
&  \sigma_{BE}^{1/2}\sigma_{B}^{-1/2}U_{B}\rho_{B}U_{B}^{\dag}\sigma
_{B}^{-1/2}\sigma_{BE}^{1/2}\nonumber\\
&  =\left(  W\left(  \sigma_{S}\otimes\left\vert 0\right\rangle \left\langle
0\right\vert _{E^{\prime}}\right)  W^{\dag}\right)  ^{1/2}\omega_{B}\left(
W\left(  \sigma_{S}\otimes\left\vert 0\right\rangle \left\langle 0\right\vert
_{E^{\prime}}\right)  W^{\dag}\right)  ^{1/2}\\
&  =W\left(  \sigma_{S}\otimes\left\vert 0\right\rangle \left\langle
0\right\vert _{E^{\prime}}\right)  ^{1/2}W^{\dag}\omega_{B}W\left(  \sigma
_{S}\otimes\left\vert 0\right\rangle \left\langle 0\right\vert _{E^{\prime}%
}\right)  ^{1/2}W^{\dag}\\
&  =W\left(  \sigma_{S}^{1/2}\otimes\left\vert 0\right\rangle \left\langle
0\right\vert _{E^{\prime}}\right)  W^{\dag}\omega_{B}W\left(  \sigma_{S}%
^{1/2}\otimes\left\vert 0\right\rangle \left\langle 0\right\vert _{E^{\prime}%
}\right)  W^{\dag}\\
&  =W\left(  \sigma_{S}^{1/2}\otimes\left\vert 0\right\rangle \left\langle
0\right\vert _{E^{\prime}}\right)  W^{\dag}\left[  \omega_{B}\otimes
I_{E}\right]  W\left(  \sigma_{S}^{1/2}\otimes\left\vert 0\right\rangle
\left\langle 0\right\vert _{E^{\prime}}\right)  W^{\dag}.
\end{align}
Continuing, the last line above is equal to%
\begin{align}
&  W\left(  \sigma_{S}^{1/2}\otimes\left\vert 0\right\rangle \left\langle
0\right\vert _{E^{\prime}}\right)  W^{\dag}\left[  \omega_{B}\otimes\sum
_{i}\left\vert i\right\rangle \left\langle i\right\vert _{E}\right]  W\left(
\sigma_{S}^{1/2}\otimes\left\vert 0\right\rangle \left\langle 0\right\vert
_{E^{\prime}}\right)  W^{\dag}\\
&  =W\left[  \left(  \sigma_{S}^{1/2}\left[  \sum_{i}\left\langle 0\right\vert
_{E^{\prime}}W^{\dag}\left\vert i\right\rangle _{E}\ \omega_{B}\ \left\langle
i\right\vert _{E}W\left\vert 0\right\rangle _{E^{\prime}}\right]  \sigma
_{S}^{1/2}\right)  \otimes\left\vert 0\right\rangle \left\langle 0\right\vert
_{E^{\prime}}\right]  W^{\dag}\\
&  =W\left(  \left[  \sigma_{S}^{1/2}\mathcal{N}^{\dag}\left[  \left(
\mathcal{N}_{S\rightarrow B}\left(  \sigma_{S}\right)  \right)  ^{-1/2}%
U_{B}\mathcal{N}_{S\rightarrow B}\left(  \rho_{S}\right)  U_{B}^{\dag}\left(
\mathcal{N}_{S\rightarrow B}\left(  \sigma_{S}\right)  \right)  ^{-1/2}%
\right]  \sigma_{S}^{1/2}\right]  \otimes\left\vert 0\right\rangle
\left\langle 0\right\vert _{E^{\prime}}\right)  W^{\dag} .\label{eq:almost-done}%
\end{align}
The Petz recovery map is defined as%
\begin{equation}
\mathcal{R}_{\sigma,\mathcal{N}}\left(  \cdot\right)  \equiv\sigma_{S}%
^{1/2}\mathcal{N}^{\dag}\left[  \left(  \mathcal{N}_{S\rightarrow B}\left(
\sigma_{S}\right)  \right)  ^{-1/2}\left(  \cdot\right)  \left(
\mathcal{N}_{S\rightarrow B}\left(  \sigma_{S}\right)  \right)  ^{-1/2}%
\right]  \sigma_{S}^{1/2}.
\end{equation}
Then by inspection, (\ref{eq:almost-done}) is equal to%
\begin{equation}
W\left(  \left[  \mathcal{R}_{\sigma,\mathcal{N}}\left(  U_{B}\mathcal{N}%
_{S\rightarrow B}\left(  \rho_{S}\right)  U_{B}^{\dag}\right)  \right]
\otimes\left\vert 0\right\rangle \left\langle 0\right\vert _{E^{\prime}%
}\right)  W^{\dag}.
\end{equation}
So the fidelity in the remainder term of (\ref{eq:remainder-term}) is%
\begin{align}
&  F\left(  \rho_{BE},V_{BE}W\left(  \left[  \mathcal{R}_{\sigma,\mathcal{N}%
}\left(  U_{B}\mathcal{N}\left(  \rho\right)  U_{B}^{\dag}\right)  \right]
\otimes\left\vert 0\right\rangle \left\langle 0\right\vert _{E^{\prime}%
}\right)  W^{\dag}\left(  V_{BE}\right)  ^{\dag}\right)  \nonumber\\
&  =F\left(  W\left(  \rho_{S}\otimes\left\vert 0\right\rangle \left\langle
0\right\vert _{E^{\prime}}\right)  W^{\dag},V_{BE}W\left(  \left[
\mathcal{R}_{\sigma,\mathcal{N}}\left(  U_{B}\mathcal{N}\left(  \rho
_{S}\right)  U_{B}^{\dag}\right)  \right]  \otimes\left\vert 0\right\rangle
\left\langle 0\right\vert _{E^{\prime}}\right)  W^{\dag}\left(  V_{BE}\right)
^{\dag}\right)  \\
&  =F\left(  \rho_{S},\left\langle 0\right\vert _{E^{\prime}}W^{\dag}%
V_{BE}W\left(  \left[  \mathcal{R}_{\sigma,\mathcal{N}}\left(  U_{B}%
\mathcal{N}\left(  \rho_{S}\right)  U_{B}^{\dag}\right)  \right]
\otimes\left\vert 0\right\rangle \left\langle 0\right\vert _{E^{\prime}%
}\right)  W^{\dag}\left(  V_{BE}\right)  ^{\dag}W\left\vert 0\right\rangle
_{E^{\prime}}\right)  \\
&  =F\left(  \rho_{S},V_{S}\left(  \mathcal{R}_{\sigma,\mathcal{N}}\left(
U_{B}\mathcal{N}\left(  \rho_{S}\right)  U_{B}^{\dag}\right)  \right)
V_{S}^{\dag}\right)  .
\end{align}
Given that $V_{BE}$ acts only on the image of the isometry $W_{SE^{\prime
}\rightarrow BE}\left\vert 0\right\rangle _{E^{\prime}}$, the second equality
follows because in this case the fidelity is invariant under the partial
isometry $\left\langle 0\right\vert _{E^{\prime}}W^{\dag}$. The last
equality follows because we can define a unitary $V_{S}$ acting on the input
space as%
\begin{equation}
V_{S}\equiv\left\langle 0\right\vert _{E^{\prime}}W^{\dag}V_{BE}W\left\vert
0\right\rangle _{E^{\prime}}.
\end{equation}
So the final remainder term for monotonicity of relative entropy is%
\begin{equation}
D\left(  \rho_{S}\Vert\sigma_{S}\right)  -D\left(  \mathcal{N}\left(  \rho
_{S}\right)  \Vert\mathcal{N}\left(  \sigma_{S}\right)  \right)  \geq-\log
F\left(  \rho_{S},V_{S}\left(  \mathcal{R}_{\sigma,\mathcal{N}}\left(
U_{B}\mathcal{N}\left(  \rho_{S}\right)  U_{B}^{\dag}\right)  \right)
V_{S}^{\dag}\right)  .
\end{equation}

\end{proof}

\begin{remark}
\label{rmk:mono} Suppose in Theorem~\ref{cor:mono}\ that $\sigma_{S}$ is a
density operator. It remains open to quantify the performance of the rotated
Petz recovery map $\mathcal{V}_{S}\circ\mathcal{R}_{\sigma,\mathcal{N}}%
^{P}\circ\mathcal{U}_{B}$ on the state $\mathcal{N}_{S\rightarrow B}\left(
\sigma_{S}\right)  $.
\end{remark}

\section{Equivalence of relative entropy inequalities with remainder terms}

As discussed in the introduction as well as in Remarks~\ref{rmk:partial_trace}
and \ref{rmk:mono}, it would be desirable to have refinements of the
inequalities in \eqref{eq:mono-rel-ent} and
\eqref{eq:mono-partial}-\eqref{eq:concavity} in terms of the Petz recovery map
(and not merely in terms of a rotated Petz recovery map). Here, we establish
the following equivalence result, depicted in Figure~\ref{fig:circle-figure}.
The remainder terms are given in terms of the square of the Bures distance
between two density operators \cite{Bures1969}, defined as%
\begin{equation}
D_{B}^{2}\left(  \rho,\sigma\right)  \equiv2\left(  1-\sqrt{F\left(
\rho,\sigma\right)  }\right)  ,
\end{equation}
where $F\left(  \rho,\sigma\right)  $ is the quantum fidelity.

\begin{theorem}
\label{thm:circle}The following inequalities with remainder terms are
equivalent (however it is an open question to determine whether any single one
of them is true):

\begin{enumerate}
\item \textbf{Strong subadditivity of entropy}. Let $\omega_{ABC}$ be a
tripartite density operator such that $\omega_{C}$ is positive definite. Then
\begin{equation}
I(A;B|C)_{\omega}\geq D_{B}^{2}\left(  \omega_{ABC},\mathcal{R}_{C\rightarrow
AC}^{P}(\omega_{BC})\right)  ,
\end{equation}
where $\mathcal{R}_{C\rightarrow AC}^{P}(\cdot)\equiv\omega_{AC}^{1/2}%
\omega_{C}^{-1/2}(\cdot)\omega_{C}^{-1/2}\omega_{AC}^{1/2}$ denotes the Petz
recovery channel.

\item \textbf{Concavity of conditional entropy}. Let $p_{X}\left(  x\right)  $
be a probability distribution characterizing the ensemble $\left\{
p_{X}\left(  x\right)  ,\rho_{AB}^{x}\right\}  $ with bipartite density
operators $\rho_{AB}^{x}$. Let $\overline{\rho}_{AB}\equiv\sum_{x}p_{X}%
(x)\rho_{AB}^{x}$ such that $\overline{\rho}_{B}$ is positive definite. Then
\begin{equation}
H(A|B)_{\overline{\rho}}-\sum_{x}p_{X}(x)H(A|B)_{\rho^{x}}\geq\sum_{x}%
p_{X}(x)D_{B}^{2}(\rho_{AB}^{x},\overline{\rho}_{AB}^{1/2}\overline{\rho}%
_{B}^{-1/2}\rho_{B}^{x}\overline{\rho}_{B}^{-1/2}\overline{\rho}_{AB}^{1/2}).
\label{eq:concavity-remainder-1}%
\end{equation}

\item \textbf{Monotonicity of relative entropy with respect to partial trace}. Let
$\rho_{AB}$ and $\sigma_{AB}$ be bipartite density operators such that
$\operatorname{supp}(\rho_{AB})\subseteq\operatorname{supp}(\sigma_{AB})$ and
$\sigma_{B}$ is positive definite. Then
\begin{equation}
D(\rho_{AB}\Vert\sigma_{AB})-D(\rho_{B}\Vert\sigma_{B})\geq D_{B}^{2}\left(
\rho_{AB},\mathcal{R}_{\sigma,\mathrm{Tr}_{A}}^{P}(\rho_{B})\right)  ,
\end{equation}
where $\mathcal{R}_{\sigma,\mathrm{Tr}_{A}}^{P}(\cdot)\equiv\sigma_{AB}%
^{1/2}\sigma_{B}^{-1/2}(\cdot)\sigma_{B}^{-1/2}\sigma_{AB}^{1/2}$ denotes the
Petz recovery channel with respect to $\sigma_{AB}$ and $\mathrm{Tr}_{A}$.

\item \textbf{Joint convexity of relative entropy}. Let $p_{X}\left(
x\right)  $ be a probability distribution characterizing the ensembles
$\left\{  p_{X}\left(  x\right)  ,\rho_{x}\right\}  $, and $\left\{
p_{X}\left(  x\right)  ,\sigma_{x}\right\}  $ with $\rho_{x}$ and $\sigma_{x}$
density operators such that $\operatorname{supp}(\rho_{x})\subseteq
\operatorname{supp}(\sigma_{x})$. Let $\overline{\rho}\equiv\sum_{x}%
p_{X}\left(  x\right)  \rho_{x}$ and $\overline{\sigma}\equiv\sum_{x}%
p_{X}\left(  x\right)  \sigma_{x}$ such that $\overline{\sigma}$ is positive
definite. Then
\begin{equation}
\sum_{x}p_{X}\left(  x\right)  D\left(  \rho_{x}\Vert\sigma_{x}\right)
-D\left(  \overline{\rho}\Vert\overline{\sigma}\right)  \geq\sum_{x}%
p_{X}\left(  x\right)  D_{B}^{2}\left(  \rho_{x},\sigma_{x}^{1/2}\left(
\overline{\sigma}\right)  ^{-1/2}\overline{\rho}\left(  \overline{\sigma
}\right)  ^{-1/2}\sigma_{x}^{1/2}\right)  .
\end{equation}

\item \textbf{Monotonicity of relative entropy}. Let $\rho$ and $\sigma$ be
density operators such that $\operatorname{supp}(\rho)\subseteq
\operatorname{supp}(\sigma)$, and $\mathcal{N}$ a CPTP map such that
$\mathcal{N}(\sigma)$ is positive definite. Then
\begin{equation}
D(\rho\Vert\sigma)-D(\mathcal{N}(\rho)\Vert\mathcal{N}(\sigma))\geq D_{B}%
^{2}\left(  \rho,\mathcal{R}_{\sigma,\mathcal{N}}^{P}(\rho)\right)  ,
\end{equation}
where $\mathcal{R}_{\sigma,\mathcal{N}}^{P}(\cdot)\equiv\sigma^{1/2}%
\mathcal{N}^{\dag}\left(  \left[  \mathcal{N}\left(  \sigma\right)  \right]
^{-1/2}(\cdot)\left[  \mathcal{N}\left(  \sigma\right)  \right]
^{-1/2}\right)  \sigma^{1/2}$ denotes the Petz recovery channel with respect
to $\sigma$ and $\mathcal{N}$.
\end{enumerate}
\end{theorem}

%

\begin{figure}
[ptb]
\begin{center}
\includegraphics[
natheight=4.706300in,
natwidth=8.553000in,
height=2.7899in,
width=5.047in
]%
{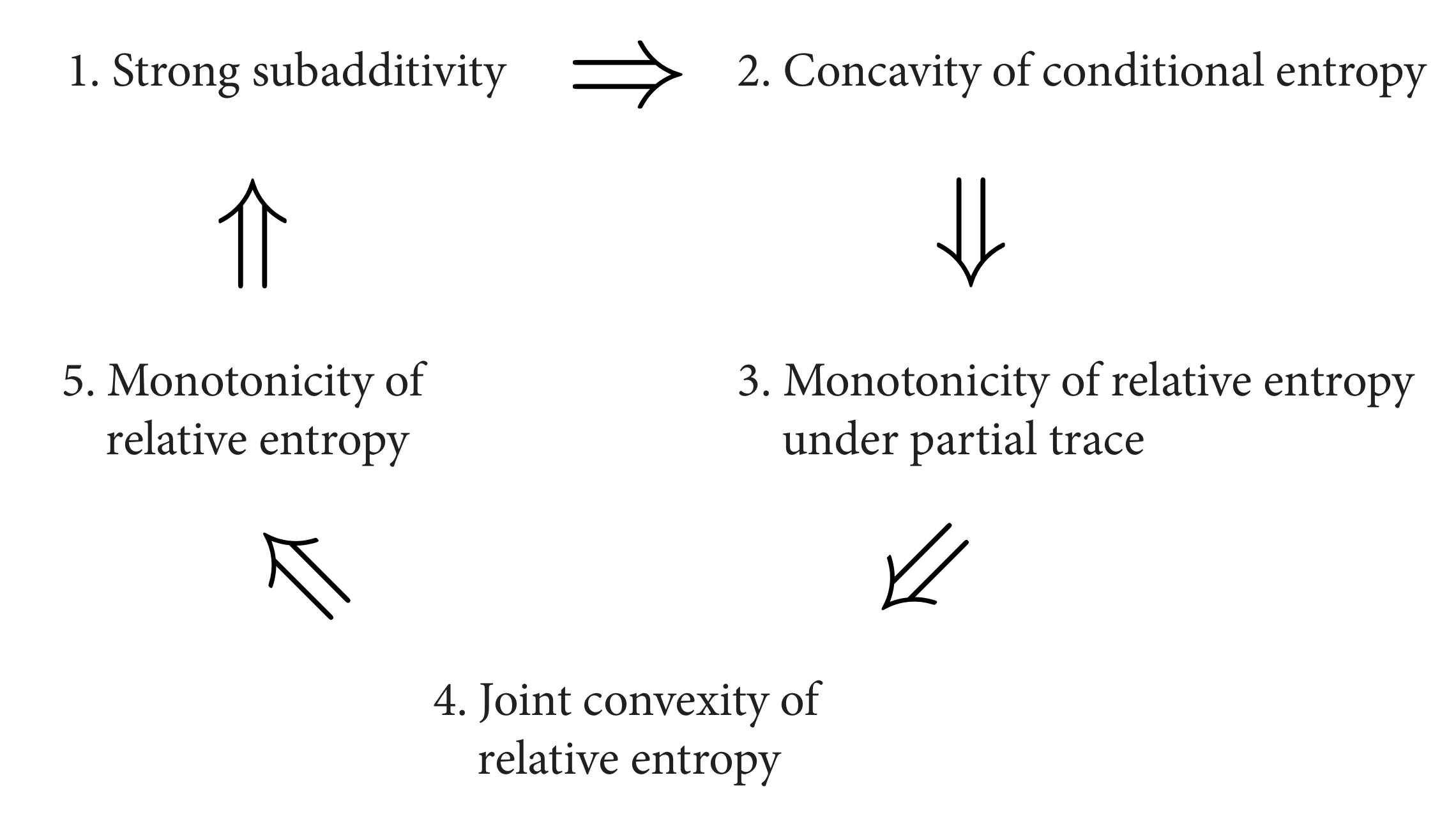}%
\caption{It is well known that all of the above fundamental entropy
inequalities are equivalent (see, e.g., \cite{R02}). Theorem~\ref{thm:circle}
extends this circle of equivalences to apply to refinements of these
inequalities in terms of the Petz recovery map.}%
\label{fig:circle-figure}%
\end{center}
\end{figure}

\begin{proof}
For the proof, we abbreviate the square root of the fidelity $F$ as the root
fidelity $\sqrt{F}$. We can easily see that 5 $\Rightarrow$ 3, and from a
variation of the development in \cite[Consequence 28]{SBW14}, we obtain 3
$\Rightarrow$ 4 $\Rightarrow$ 5, leading to 3 $\Leftrightarrow$ 4
$\Leftrightarrow$ 5.\footnote{Note that \cite[Consequence 28]{SBW14}
establishes the circle 3 $\Leftrightarrow$ 4 $\Leftrightarrow$ 5 with a
remainder term of $-\log F$.} We can get 5 $\Rightarrow$ 1
by choosing $\rho=\omega_{ABC}$, $\sigma=\omega_{AC}\otimes\omega_{B}$, and
$\mathcal{N}=\operatorname{Tr}_{A}$, so that
\begin{align}
&  \sigma^{1/2}\mathcal{N}^{\dag}\left(  \left[  \mathcal{N}\left(
\sigma\right)  \right]  ^{-1/2}(\cdot)\left[  \mathcal{N}\left(
\sigma\right)  \right]  ^{-1/2}\right)  \sigma^{1/2}\nonumber\\
&  =\left[  \omega_{AC}\otimes\omega_{B}\right]  ^{1/2}\left[  \left(  \left[
\omega_{C}\otimes\omega_{B}\right]  ^{-1/2}(\cdot)\left[  \omega_{C}%
\otimes\omega_{B}\right]  ^{-1/2}\right)  \otimes I_{A}\right]  \left[
\omega_{AC}\otimes\omega_{B}\right]  ^{1/2}\\
&  =\omega_{AC}^{1/2}\omega_{C}^{-1/2}(\cdot)\omega_{C}^{-1/2}\omega
_{AC}^{1/2}.
\end{align}
Then%
\begin{align}
I\left(  A;B|C\right)  _{\omega} &  =D(\omega_{ABC}\Vert\omega_{AC}%
\otimes\omega_{B})-D(\omega_{BC}\Vert\omega_{C}\otimes\omega_{B})\\
&  \geq2\left(  1-\sqrt{F}\left(  \omega_{ABC},\mathcal{R}_{\sigma
,\mathcal{N}}^{P}(\omega_{BC})\right)  \right)  \\
&  =2\left(  1-\sqrt{F}\left(  \omega_{ABC},\omega_{AC}^{1/2}\omega_{C}%
^{-1/2}\omega_{BC}\omega_{C}^{-1/2}\omega_{AC}^{1/2}\right)  \right)  .
\end{align}
The implication 1 $\Rightarrow$ 2 follows by choosing%
\begin{equation}
\theta_{XAB}\equiv\sum_{x}p_{X}\left(  x\right)  \left\vert x\right\rangle
\left\langle x\right\vert _{X}\otimes\rho_{AB}^{x},
\end{equation}
so that%
\begin{align}
H(A|B)_{\overline{\rho}}-\sum_{x}p_{X}(x)H(A|B)_{\rho^{x}} &  =I\left(
A;X|B\right)  _{\theta}\\
&  \geq2\left(  1-\sqrt{F}\left(  \theta_{XAB},\theta_{AB}^{1/2}\theta
_{B}^{-1/2}\theta_{XB}\theta_{B}^{-1/2}\theta_{AB}^{1/2}\right)  \right)  \\
&  =2\left(  1-\sum_{x}p_{X}(x)\sqrt{F}\left(  \rho_{AB}^{x},\overline{\rho
}_{AB}^{1/2}\overline{\rho}_{B}^{-1/2}\rho_{B}^{x}\overline{\rho}_{B}%
^{-1/2}\overline{\rho}_{AB}^{1/2}\right)  \right) . \
\end{align}

The last remaining implication 2 $\Rightarrow$ 3 has the most involved proof,
which we establish now by using the idea from \cite[Section~3-E]{LR73}.
Throughout our proof, we employ Theorem~V.3.3 of \cite{B97}. This theorem states that if
$f$ is a differentiable function on an open neighborhood of the spectrum of
some self-adjoint operator $A$, then its derivative $Df$ at $A$ is given by%
\begin{equation}
Df\left(  A\right)  :H\rightarrow\sum_{\lambda,\eta}f^{\left[  1\right]
}\left(  \lambda,\eta\right)  P_{A}\left(  \lambda\right)  HP_{A}\left(
\eta\right)  ,
\end{equation}
where $A=\sum_{\lambda}\lambda P_{A}\left(  \lambda\right)  $ is the spectral
decomposition of $A$, and $f^{\left[  1\right]  }$ is the first divided
difference function. In particular, if $x\longmapsto A\left(  x\right)
\in\mathcal{B}\left(  \mathcal{H}\right)  _{+}$ is a differentiable function
on an open interval in $\mathbb{R}$, with derivative $A^{\prime}$, then%
\begin{equation}
\frac{d}{dx}f\left(  A\left(  x\right)  \right)  =\sum_{\lambda,\eta
}f^{\left[  1\right]  }\left(  \lambda,\eta\right)  P_{A\left(  x\right)
}\left(  \lambda\right)  A^{\prime}\left(  x\right)  P_{A\left(  x\right)
}\left(  \eta\right)  ,\label{eq:op-deriv}%
\end{equation}
so that%
\begin{equation}
\frac{d}{dx}\text{Tr}\left\{  f\left(  A\left(  x\right)  \right)  \right\}
=\text{Tr}\left\{  f^{\prime}\left(  A\left(  x\right)  \right)  A^{\prime
}\left(  x\right)  \right\}  .
\end{equation}
In particular, if $A\left(  x\right)  =A+xB$, then%
\begin{equation}
\frac{d}{dx}\text{Tr}\left\{  f\left(  A\left(  x\right)  \right)  \right\}
=\text{Tr}\left\{  f^{\prime}\left(  A\left(  x\right)  \right)  B\right\}
.\label{eq:critical-one}%
\end{equation}
We can now proceed. In what follows, we will be taking $A\left(  x\right)
=\sigma_{AB}+x\rho_{AB}$, where $\sigma_{AB}$ is a positive definite density
operator, $\rho_{AB}$ is a density operator, and $x\geq0$. We also make use of
the standard fact that the function $f:X\rightarrow X^{-1}$ is everywhere
differentiable on the set of invertible density operators, and at an
invertible $X$, its derivative is $f^{\prime}\left(  X\right)  :Y\rightarrow
-X^{-1}YX^{-1}$.

Consider that the conditional entropy is homogeneous, in the sense that%
\begin{equation}
H\left(  A|B\right)  _{xG}=xH\left(  A|B\right)  _{G},
\end{equation}
where $x$ is a positive scalar and $G_{AB}$ is a positive semi-definite
operator on systems $AB$. Let%
\begin{equation}
\xi_{YAB}\equiv\frac{1}{x+1}\left\vert 0\right\rangle \left\langle
0\right\vert _{Y}\otimes\sigma_{AB}+\frac{x}{x+1}\left\vert 1\right\rangle
\left\langle 1\right\vert _{Y}\otimes\rho_{AB},
\end{equation}
with $\sigma_{AB}$ a positive definite density operator and $\rho_{AB}$ a
density operator. Then it follows from homogeneity and concavity with the Petz
remainder term (by assumption) that%
\begin{align}
H\left(  A|B\right)  _{\sigma+x\rho} &  =\left(  x+1\right)  H\left(
A|B\right)  _{\xi}\\
&  \geq\left(  x+1\right)  \left[  \frac{1}{x+1}H\left(  A|B\right)  _{\sigma
}+\frac{x}{x+1}H\left(  A|B\right)  _{\rho}+R\left(  x,\sigma_{AB},\rho
_{AB}\right)  \right]  \\
&  =H\left(  A|B\right)  _{\sigma}+xH\left(  A|B\right)  _{\rho}+\left(
x+1\right)  R\left(  x,\sigma_{AB},\rho_{AB}\right)  ,
\end{align}
where%
\begin{multline}
R\left(  x,\sigma_{AB},\rho_{AB}\right)  \equiv\\
2\left(  1-\left[  \frac{1}{x+1}\sqrt{F}\left(  \sigma_{AB},\xi_{AB}^{1/2}%
\xi_{B}^{-1/2}\sigma_{B}\xi_{B}^{-1/2}\xi_{AB}^{1/2}\right)  +\frac{x}%
{x+1}\sqrt{F}\left(  \rho_{AB},\xi_{AB}^{1/2}\xi_{B}^{-1/2}\rho_{B}\xi
_{B}^{-1/2}\xi_{AB}^{1/2}\right)  \right]  \right)  .
\end{multline}
Manipulating the above inequality then gives%
\begin{equation}
\frac{H\left(  A|B\right)  _{\sigma+x\rho}-H\left(  A|B\right)  _{\sigma}}%
{x}\geq H\left(  A|B\right)  _{\rho}+\frac{x+1}{x}R\left(  x,\sigma_{AB}%
,\rho_{AB}\right)  .
\end{equation}
Taking the limit as $x\searrow0$ then gives%
\begin{equation}
\lim_{x\searrow0}\frac{H\left(  A|B\right)  _{\sigma+x\rho}-H\left(
A|B\right)  _{\sigma}}{x}=\left.  \frac{d}{dx}H\left(  A|B\right)
_{\sigma+x\rho}\right\vert _{x=0}\geq H\left(  A|B\right)  _{\rho}%
+\lim_{x\searrow0}\frac{x+1}{x}R\left(  x,\sigma_{AB},\rho_{AB}\right)
.\label{eq:lim}%
\end{equation}
We now evaluate the limits separately, beginning with the one on the left hand
side. So we consider
\begin{equation}
\frac{d}{dx}H\left(  A|B\right)  _{\sigma+x\rho}=\frac{d}{dx}\left[
-\operatorname{Tr}\left\{  \left(  \sigma_{AB}+x\rho_{AB}\right)  \log\left(
\sigma_{AB}+x\rho_{AB}\right)  \right\}  +\operatorname{Tr}\left\{  \left(
\sigma_{B}+x\rho_{B}\right)  \log\left(  \sigma_{B}+x\rho_{B}\right)
\right\}  \right]  .
\end{equation}
We evaluate this by using $\frac{d}{dy}\left[  g\left(  y\right)  \log
g\left(  y\right)  \right]  =\left[  \log g\left(  y\right)  +1\right]
g^{\prime}\left(  y\right)  $ and (\ref{eq:critical-one}) to find that%
\begin{equation}
\frac{d}{dx}\operatorname{Tr}\left\{  \left(  \sigma_{AB}+x\rho_{AB}\right)
\log\left(  \sigma_{AB}+x\rho_{AB}\right)  \right\}  =\operatorname{Tr}%
\left\{  \left[  \log\left(  \sigma_{AB}+x\rho_{AB}\right)  +I_{AB}\right]
\rho_{AB}\right\}  ,
\end{equation}
so that%
\begin{equation}
\frac{d}{dx}H\left(  A|B\right)  _{\sigma+x\rho}=-\operatorname{Tr}\left\{
\rho_{AB}\log\left(  \sigma_{AB}+x\rho_{AB}\right)  \right\}
+\operatorname{Tr}\left\{  \rho_{B}\log\left(  \sigma_{B}+x\rho_{B}\right)
\right\}  ,
\end{equation}
and thus%
\begin{equation}
\left.  \frac{d}{dx}H\left(  A|B\right)  _{\sigma+x\rho}\right\vert
_{x=0}=-\operatorname{Tr}\left\{  \rho_{AB}\log\sigma_{AB}\right\}
+\operatorname{Tr}\left\{  \rho_{B}\log\sigma_{B}\right\}  .
\end{equation}
Substituting back into the inequality~\eqref{eq:lim}, we find that%
\begin{multline}
-\operatorname{Tr}\left\{  \rho_{AB}\log\sigma_{AB}\right\}
+\operatorname{Tr}\left\{  \rho_{B}\log\sigma_{B}\right\}  \geq\\
-\operatorname{Tr}\left\{  \rho_{AB}\log\rho_{AB}\right\}  +\operatorname{Tr}%
\left\{  \rho_{B}\log\rho_{B}\right\}  +\lim_{x\searrow0}\frac{x+1}{x}R\left(
x,\sigma_{AB},\rho_{AB}\right)  ,
\end{multline}
which is equivalent to (cf., \cite[Eq.~(3.2)]{LR73})%
\begin{equation}
D\left(  \rho_{AB}\Vert\sigma_{AB}\right)  -D\left(  \rho_{B}\Vert\sigma
_{B}\right)  \geq\lim_{x\searrow0}\frac{x+1}{x}R\left(  x,\sigma_{AB}%
,\rho_{AB}\right)  .
\end{equation}
So we need to evaluate this last limit to get the remainder term. Consider
that%
\begin{align}
&  \lim_{x\searrow0}\frac{x+1}{x}R\left(  x,\sigma_{AB},\rho_{AB}\right)  \\
&  =\lim_{x\searrow0}2\left(  1+\frac{1-\sqrt{F}\left(  \sigma_{AB},\xi
_{AB}^{1/2}\xi_{B}^{-1/2}\sigma_{B}\xi_{B}^{-1/2}\xi_{AB}^{1/2}\right)  }%
{x}-\sqrt{F}\left(  \rho_{AB},\xi_{AB}^{1/2}\xi_{B}^{-1/2}\rho_{B}\xi
_{B}^{-1/2}\xi_{AB}^{1/2}\right)  \right)  .
\end{align}
Since%
\begin{equation}
\lim_{x\searrow0}\sqrt{F}\left(  \rho_{AB},\xi_{AB}^{1/2}\xi_{B}^{-1/2}%
\rho_{B}\xi_{B}^{-1/2}\xi_{AB}^{1/2}\right)  =\sqrt{F}\left(  \rho_{AB}%
,\sigma_{AB}^{1/2}\sigma_{B}^{-1/2}\rho_{B}\sigma_{B}^{-1/2}\sigma_{AB}%
^{1/2}\right)  ,
\end{equation}
it remains to show that%
\begin{multline}
\lim_{x\searrow0}\frac{1-\sqrt{F}\left(  \sigma_{AB},\xi_{AB}^{1/2}\xi
_{B}^{-1/2}\sigma_{B}\xi_{B}^{-1/2}\xi_{AB}^{1/2}\right)  }{x}\\
=\left.  \frac{d}{dx}\sqrt{F}\left(  \sigma_{AB},\xi_{AB}^{1/2}\xi_{B}%
^{-1/2}\sigma_{B}\xi_{B}^{-1/2}\xi_{AB}^{1/2}\right)  \right\vert _{x=0}=0.
\end{multline}
Essentially, this derivative vanishes because the fidelity
is one at $x=0$ and therefore maximal. In what follows, we explicitly show that the derivative above is equal to
zero. Consider that%
\begin{align}
&  \sqrt{F}\left(  \sigma_{AB},\xi_{AB}^{1/2}\xi_{B}^{-1/2}\sigma_{B}\xi
_{B}^{-1/2}\xi_{AB}^{1/2}\right)  \nonumber\\
&  =\operatorname{Tr}\left\{  \left(  \sigma_{AB}^{1/2}\xi_{AB}^{1/2}\xi
_{B}^{-1/2}\sigma_{B}\xi_{B}^{-1/2}\xi_{AB}^{1/2}\sigma_{AB}^{1/2}\right)
^{1/2}\right\}  \\
&  =\operatorname{Tr}\left\{  \left(  \sigma_{AB}^{1/2}\left(  \sigma
_{AB}+x\rho_{AB}\right)  ^{1/2}\left(  \sigma_{B}+x\rho_{B}\right)
^{-1/2}\sigma_{B}\left(  \sigma_{B}+x\rho_{B}\right)  ^{-1/2}\left(
\sigma_{AB}+x\rho_{AB}\right)  ^{1/2}\sigma_{AB}^{1/2}\right)  ^{1/2}\right\}
,\label{eq:fid-exp}%
\end{align}
as well as%
\begin{equation}
\frac{d}{dx}\operatorname{Tr}\left\{  \left(  G\left(  x\right)  \right)
^{1/2}\right\}  =\frac{1}{2}\operatorname{Tr}\left\{  G\left(  x\right)
^{-1/2}\frac{d}{dx}G\left(  x\right)  \right\}  ,
\end{equation}
which follows from (\ref{eq:critical-one}). Applying the above rule, we get
that $\frac{d}{dx}$ of (\ref{eq:fid-exp}) is equal to%
\begin{equation}
\operatorname{Tr}\left\{
\begin{array}
[c]{c}%
\left(  \sigma_{AB}^{1/2}\left(  \sigma_{AB}+x\rho_{AB}\right)  ^{1/2}\left(
\sigma_{B}+x\rho_{B}\right)  ^{-1/2}\sigma_{B}\left(  \sigma_{B}+x\rho
_{B}\right)  ^{-1/2}\left(  \sigma_{AB}+x\rho_{AB}\right)  ^{1/2}\sigma
_{AB}^{1/2}\right)  ^{-1/2}\times\\
\sigma_{AB}^{1/2}\frac{d}{dx}\left[  \left(  \sigma_{AB}+x\rho_{AB}\right)
^{1/2}\left(  \sigma_{B}+x\rho_{B}\right)  ^{-1/2}\sigma_{B}\left(  \sigma
_{B}+x\rho_{B}\right)  ^{-1/2}\left(  \sigma_{AB}+x\rho_{AB}\right)
^{1/2}\right]  \sigma_{AB}^{1/2}%
\end{array}
\right\}  .\label{eq:big-trace}%
\end{equation}
Now, take the limit as $x\searrow0$ to find that (\ref{eq:big-trace}) is equal
to
\begin{align}
&  \left.  \frac{d}{dx}\sqrt{F}\left(  \sigma_{AB},\xi_{AB}^{1/2}\xi
_{B}^{-1/2}\sigma_{B}\xi_{B}^{-1/2}\xi_{AB}^{1/2}\right)  \right\vert
_{x=0}\nonumber\\
&  =\operatorname{Tr}\left\{
\begin{array}
[c]{c}%
\left(  \sigma_{AB}^{1/2}\sigma_{AB}^{1/2}\sigma_{B}^{-1/2}\sigma_{B}%
\sigma_{B}^{-1/2}\sigma_{AB}^{1/2}\sigma_{AB}^{1/2}\right)  ^{-1/2}\times\\
\sigma_{AB}^{\frac{1}{2}}\left.  \frac{d}{dx}\left[  \left(  \sigma_{AB}%
+x\rho_{AB}\right)  ^{\frac{1}{2}}\left(  \sigma_{B}+x\rho_{B}\right)
^{-\frac{1}{2}}\sigma_{B}\left(  \sigma_{B}+x\rho_{B}\right)  ^{-\frac{1}{2}%
}\left(  \sigma_{AB}+x\rho_{AB}\right)  ^{\frac{1}{2}}\right]  \right\vert
_{x=0}\sigma_{AB}^{\frac{1}{2}}%
\end{array}
\right\}  \\
&  =\operatorname{Tr}\left\{
\begin{array}
[c]{c}%
\left(  \sigma_{AB}\right)  ^{-1}\times\\
\sigma_{AB}^{\frac{1}{2}}\left.  \frac{d}{dx}\left[  \left(  \sigma_{AB}%
+x\rho_{AB}\right)  ^{\frac{1}{2}}\left(  \sigma_{B}+x\rho_{B}\right)
^{-\frac{1}{2}}\sigma_{B}\left(  \sigma_{B}+x\rho_{B}\right)  ^{-\frac{1}{2}%
}\left(  \sigma_{AB}+x\rho_{AB}\right)  ^{\frac{1}{2}}\right]  \right\vert
_{x=0}\sigma_{AB}^{\frac{1}{2}}%
\end{array}
\right\}  \\
&  =\operatorname{Tr}\left\{  \left.  \frac{d}{dx}\left[  \left(  \sigma
_{AB}+x\rho_{AB}\right)  ^{1/2}\left(  \sigma_{B}+x\rho_{B}\right)
^{-1/2}\sigma_{B}\left(  \sigma_{B}+x\rho_{B}\right)  ^{-1/2}\left(
\sigma_{AB}+x\rho_{AB}\right)  ^{1/2}\right]  \right\vert _{x=0}\right\}
\end{align}
So we focus on this last expression and note from the derivative product rule
that there are four terms to consider. We consider one at a time, beginning
with the first term:%
\begin{align}
&  \lim_{x\searrow0}\operatorname{Tr}\left\{  \frac{d}{dx}\left[  \left(
\sigma_{AB}+x\rho_{AB}\right)  ^{1/2}\right]  \left(  \sigma_{B}+x\rho
_{B}\right)  ^{-1/2}\sigma_{B}\left(  \sigma_{B}+x\rho_{B}\right)
^{-1/2}\left(  \sigma_{AB}+x\rho_{AB}\right)  ^{1/2}\right\}  \nonumber\\
&  =\operatorname{Tr}\left\{  \left.  \frac{d}{dx}\left[  \left(  \sigma
_{AB}+x\rho_{AB}\right)  ^{1/2}\right]  \right\vert _{x=0}\sigma_{B}%
^{-1/2}\sigma_{B}\sigma_{B}^{-1/2}\sigma_{AB}^{1/2}\right\}  \\
&  =\operatorname{Tr}\left\{  \left.  \frac{d}{dx}\left[  \left(  \sigma
_{AB}+x\rho_{AB}\right)  ^{1/2}\right]  \right\vert _{x=0}\sigma_{AB}%
^{1/2}\right\}  \\
&  =\frac{1}{2}\operatorname{Tr}\left\{  \rho_{AB}\right\}  \\
&  =\frac{1}{2},
\end{align}
where the second to last line follows from (\ref{eq:op-deriv}). We now
consider the second term:%
\begin{align}
&  \lim_{x\searrow0}\operatorname{Tr}\left\{  \left(  \sigma_{AB}+x\rho
_{AB}\right)  ^{1/2}\frac{d}{dx}\left[  \left(  \sigma_{B}+x\rho_{B}\right)
^{-1/2}\right]  \sigma_{B}\left(  \sigma_{B}+x\rho_{B}\right)  ^{-1/2}\left(
\sigma_{AB}+x\rho_{AB}\right)  ^{1/2}\right\}  \nonumber\\
&  =\operatorname{Tr}\left\{  \sigma_{AB}^{1/2}\left.  \frac{d}{dx}\left[
\left(  \sigma_{B}+x\rho_{B}\right)  ^{-1/2}\right]  \right\vert _{x=0}%
\sigma_{B}\sigma_{B}^{-1/2}\sigma_{AB}^{1/2}\right\}  \\
&  =\operatorname{Tr}\left\{  \sigma_{AB}\left.  \frac{d}{dx}\left[  \left(
\sigma_{B}+x\rho_{B}\right)  ^{-1/2}\right]  \right\vert _{x=0}\sigma
_{B}\sigma_{B}^{-1/2}\right\}  \\
&  =\operatorname{Tr}\left\{  \sigma_{B}\left.  \frac{d}{dx}\left[  \left(
\sigma_{B}+x\rho_{B}\right)  ^{-1/2}\right]  \right\vert _{x=0}\sigma
_{B}\sigma_{B}^{-1/2}\right\}  \\
&  =\operatorname{Tr}\left\{  \left.  \frac{d}{dx}\left[  \left(  \sigma
_{B}+x\rho_{B}\right)  ^{-1/2}\right]  \right\vert _{x=0}\sigma_{B}%
^{3/2}\right\}  \\
&  =-\frac{1}{2}\operatorname{Tr}\left\{  \sigma_{B}^{-3/2}\sigma_{B}%
^{3/2}\rho_{B}\right\}  \\
&  =-\frac{1}{2}\operatorname{Tr}\left\{  \rho_{B}\right\}  \\
&  =-\frac{1}{2}.
\end{align}
The third to last line follows from (\ref{eq:op-deriv}). Combining these
results and using that the last two terms resulting from the product rule are
Hermitian conjugates of the first two, we find that
\begin{equation}
\operatorname{Tr}\left\{  \left.  \frac{d}{dx}\left[  \left(  \sigma
_{AB}+x\rho_{AB}\right)  ^{1/2}\left(  \sigma_{B}+x\rho_{B}\right)
^{-1/2}\sigma_{B}\left(  \sigma_{B}+x\rho_{B}\right)  ^{-1/2}\left(
\sigma_{AB}+x\rho_{AB}\right)  ^{1/2}\right]  \right\vert _{x=0}\right\}  =0,
\end{equation}
which allows us to conclude that%
\begin{equation}
\left.  \frac{d}{dx}\sqrt{F}\left(  \sigma_{AB},\xi_{AB}^{1/2}\xi_{B}%
^{-1/2}\sigma_{B}\xi_{B}^{-1/2}\xi_{AB}^{1/2}\right)  \right\vert _{x=0}=0,
\end{equation}
Hence, we can conclude that the following inequality is a consequence of
(\ref{eq:concavity-remainder-1}):%
\begin{equation}
D\left(  \rho_{AB}\Vert\sigma_{AB}\right)  -D\left(  \rho_{B}\Vert\sigma
_{B}\right)  \geq2\left(  1-\sqrt{F}\left(  \rho_{AB},\sigma_{AB}^{1/2}%
\sigma_{B}^{-1/2}\rho_{B}\sigma_{B}^{-1/2}\sigma_{AB}^{1/2}\right)  \right)  .
\end{equation}

\end{proof}

\bigskip

\textbf{Note:} After the completion of the present paper,
the works in \cite{W15} and \cite{STH15} appeared, which build upon ideas established in this paper. The main contribution of \cite{W15} is to show that the rotated Petz map in Corollary~\ref{cor:mono} can take a more particular form. Specifically, the unitary channel $\mathcal{U}_B$ in Corollary~\ref{cor:mono} can be taken to commute with $\mathcal{N}(\sigma)$ and the unitary channel $\mathcal{V}_S$ can be taken to commute with $\sigma$. The main contribution of \cite{STH15} is to show that the fidelity remainder term in Corollary~\ref{cor:mono} can be replaced with the ``measured relative entropy'' and the rotated Petz map can be replaced with a ``twirled Petz map.'' Please refer to \cite{W15} and \cite{STH15} for more details. 

\bigskip

\textbf{Acknowledgements.} We are especially grateful to Rupert Frank for many
discussions on the topic of this paper. We thank the anonymous referees for many suggestions that helped to improve the paper. We acknowledge additional discussions
with Siddhartha Das, Nilanjana Datta, Omar Fawzi, Renato Renner, Volkher
Scholz, Kaushik P.~Seshadreesan, Marco Tomamichel, and Michael Walter. MMW
acknowledges support from startup funds from the Department of Physics and
Astronomy at LSU, the NSF\ under Award No.~CCF-1350397, and the DARPA Quiness
Program through US Army Research Office award W31P4Q-12-1-0019.

\appendix

\section{Auxiliary lemmas from \cite{FR14}}

\label{app:FR-lemmas}

In this appendix, for the convenience of the reader, we  list verbatim the relevant lemmas that we have used from \cite{FR14}.

  \begin{lemma}[Lemma 2.3 of \cite{FR14}] \label{lem_Dmapping}
    Let $\rho$ be a density operator, let $\sigma$ be a non-negative operator on the same space, and let $\{\mathcal{W}_n\}_{n \in \mathbb{N}}$ be a sequence of trace non-increasing completely positive maps on the $n$-fold tensor product of this space. If $\tr(\mathcal{W}_n(\rho^{\otimes n}))$  decreases less than exponentially in $n$, i.e., 
    \begin{align}
      \liminf_{n \to \infty} e^{\xi n} \tr\bigl(\mathcal{W}_n(\rho^{\otimes n})\bigr) > 0
    \end{align}
    for any $\xi > 0$, then
    \begin{align}
      \limsup_{n \to \infty} \frac{1}{n} D\bigl(\mathcal{W}_n(\rho^{\otimes n}) \| \mathcal{W}_n(\sigma^{\otimes n})\bigr) \leq D(\rho \| \sigma) \ .
    \end{align}
  \end{lemma}

  \begin{lemma}[Lemma 4.2 of \cite{FR14}]  \label{lem_subfidelity}
  Let $\rho_{R^n S^n}$ be a permutation-invariant non-negative operator on $(R \otimes S)^{\otimes n}$ and let $\sigma_{R S}$ be a non-negative operator on $R \otimes S$. Furthermore, let $W_{R^n}$ be a permutation-invariant operator on $R^{\otimes n}$ with $\|W_{R^n}\|_\infty \leq 1$. Then there exists a unitary $U_R$  on $R$ such that  \begin{align}  \label{eq_subfidelity}
    \sqrt{F}\bigl(\rho_{R^n S^n}, U_R^{\otimes n}  \sigma_{R S}^{\otimes n} (U_R^{\otimes n})^{\dagger}\bigr) \geq (n+1)^{-d^2} \sqrt{F}\bigl(W_{R^n} \rho_{R^n S^n} W_{R^n}^{\dagger},  \sigma_{RS}^{\otimes n} \bigr) \ ,
  \end{align}
  where $d =  \dim(R) \dim(S)^2$.
\end{lemma}

 \begin{lemma}[Lemma B.2 of \cite{FR14}]  \label{lem_DFidelity}
   For any non-negative operators $\rho$ and $\sigma$
   \begin{align}
      D(\rho \| \sigma)  \geq -2 \log_2 \frac{\sqrt{F}(\rho , \sigma)}{\tr(\rho)}  \ .
   \end{align}
 \end{lemma}

\begin{lemma}[Lemma B.6 of \cite{FR14}]   \label{lem_fidelityoperator}
  For any non-negative operators $\rho$ and $\sigma$ and any operator $W$  on the same space we have
  \begin{align}
    \sqrt{F}(\rho, W \sigma W^{\dagger}) = \sqrt{F}(W^{\dagger} \rho W, \sigma) \ .
  \end{align}
\end{lemma}

\begin{lemma}[Lemma B.7 of \cite{FR14}]  \label{lem_fidelitydecomposition}
  Let $\rho$ and $\sigma$ be non-negative operators and let $\{W_d\}_{d \in D}$ be a family of operators such that $\sum_{d \in D} W_d = \id$. Then
  \begin{align}
    \sum_{d \in D} \sqrt{F}(W_d^{\dagger} \rho W_d, \sigma) \geq \sqrt{F}(\rho, \sigma) \ .
  \end{align}
\end{lemma}

\bibliographystyle{alpha}
\bibliography{Ref}

\end{document}